\documentclass[twocolumn]{autart} 

\usepackage{cite}
\usepackage{amsmath,amssymb,amsfonts}
\usepackage{algorithm}
\usepackage{algorithmicx}
\usepackage{algpseudocode}
\usepackage{array}
\usepackage{graphicx}
\usepackage{epstopdf}
\usepackage{enumerate}
\usepackage{wrapfig}

\makeatletter
\def\journal#1{}
\def\ps@copyright{\let\@oddfoot\@empty\let\@evenfoot\@empty}
\def\ps@firstpage{\let\@oddfoot\@empty\let\@evenfoot\@empty}
\makeatother

\newtheorem{example}{Example}
\newtheorem{assumption}{Assumption}
\newtheorem{definition}{Definition}
\newtheorem{theorem}{Theorem}

\newtheorem{proposition}{Proposition}

\enlargethispage{6\baselineskip} 
\begin{document}
	
	\begin{frontmatter}
		\title{Learning Neural Controllers for Nonlinear Systems from Data\thanksref{footnoteinfo}}

		\thanks[footnoteinfo]{This paper was not presented at any IFAC meeting. Corresponding author Zhi-Wei Liu.}
			
		\author[HUST]{Zhongjie Hu}\ead{zhongjiehu@hust.edu.cn},
		\author[HUST]{Zhi-Wei Liu}\ead{zwliu@hust.edu.cn},
		\author[HUST]{Chen Wang}\ead{wangchen\_math@hust.edu.cn}
		
		\address[HUST]{School of Artificial Intelligence and Automation, Huazhong University of Science and Technology, Wuhan 430074, China}

		\begin{keyword}
			Data-driven control, neural networks, formal verification, nonlinear control systems, robust control.
		\end{keyword}

		\begin{abstract}
			This article addresses the problem of designing neural feedback controllers for unknown nonlinear systems. We propose an indirect data-driven framework that uses offline data to identify the system dynamics, upon which a neural feedback controller and a neural Lyapunov function are jointly synthesized. Input constraints are enforced by integrating a hard-saturation structure into the controller architecture. Robust synthesis conditions are derived to account for data perturbations during identification. Formal stability is certified by combining SMT verification with local Lyapunov analysis near the equilibrium. Numerical examples validate the effectiveness of the proposed framework. 
		\end{abstract}
	\end{frontmatter}
	
	\section{Introduction}
	\label{sec:intro}

	The rapid advances in deep learning and machine learning have revolutionized many branches of science and engineering, including robotics~\cite{lee2020learning}, magnetic plasma control~\cite{degrave2022magnetic}, and power systems~\cite{wang2027model}, which has established \emph{learning for control} as a prevalent paradigm for tackling complex tasks. Methodologically, learning-based control can be categorized into \emph{indirect} methods~\cite{pillonetto2014kernel, ljung2010perspectives}, which first identify the system model to enable model-based control synthesis, and \emph{direct} methods~\cite{de2019formulas, campi2002virtual}, which learn controllers directly from data without explicitly attempting to identify the system.

	Within these frameworks, a major body of research formulates controller synthesis as convex optimization, such as Semidefinite Programming (SDP)~\cite{dprt2023cancellation, van2020noisy}, and Sum-of-Squares (SOS) formulations~\cite{guo2021data, huang2026feedback}. While offering appealing computational tractability, these formulations typically restrict candidate controllers to linear feedback or fixed-structure nonlinear laws \cite{dpt2023arc, martin2023guarantees}. Such structural constraints limit the expressive capacity of the controller, thus imposing a fundamental bottleneck on achievable closed-loop performance. Furthermore, candidate Lyapunov functions within these paradigms are predominantly confined to standard quadratic forms or fixed-degree polynomial functions \cite{bisoffi2022data, guo2024data}, which restricts the search space for valid stability certificates.

	In this paper, we are interested in leveraging neural networks to parameterize both the controller and the Lyapunov function. By the universal approximation theorem \cite{hornik1989multilayer}, neural networks are capable of approximating any continuous function with arbitrary accuracy. Employing neural architectures thus significantly enhances the expressivity of both search spaces. 
	Despite their expressive capacity, deploying neural networks in safety-critical applications raises concerns regarding reliability and formal correctness \cite{manchester2026neural, mestres2025neural}. In fact, neural networks are trained on finite, discrete data samples, while stability must hold continuously over the entire domain, which makes rigorous stability guarantees non-trivial \cite{dai2021lyapunov, abate2020formal}. In addition, stability conditions involving neural networks are inherently non-convex, rendering classical SDP or SOS optimization tools inapplicable.

	In recent years, Satisfiability Modulo Theories (SMT) has been increasingly utilized for formal stability certification \cite{edwards2025general, ahmed2020automated}. In general, SMT solvers operate by formulating the negation of desired property conditions into a first-order logical formula, and then searching for counterexamples within the domain of interest \cite{nieuwenhuis2006solving, barrett2018satisfiability}. A determination of unsatisfiability (UNSAT) guarantees the absence of counterexamples, thereby certifying that the target conditions hold across the entire domain. Representative SMT solvers include dReal \cite{gao2013dreal, gao2012delta}, Z3 \cite{de2008z3}, and Reluplex \cite{katz2017reluplex}. Building upon these solvers, the counterexample-guided inductive synthesis (CEGIS) framework, which is typically structured as a learner-falsifier loop, was developed for the automated formal synthesis of neural controllers and stability certificates \cite{abate2018counterexample, abate2021fossil, edwards2024fossil}. This paradigm was first introduced to co-synthesize neural controllers and Lyapunov functions for known nonlinear systems \cite{chang2019neural}, and was subsequently extended to unknown systems by employing Lipschitz-based error bounds to bridge neural system identification and robust SMT verification \cite{zhou2022neural}. More recently, to incorporate performance specifications beyond mere stabilization, this framework was generalized to input-output dissipativity, where neural storage functions and supply rates are learned to guarantee stability along with optimality \cite{wang2026designing}.

	\textit{Contribution:} The purpose of this paper is to design neural feedback controllers for unknown nonlinear systems with provable stability guarantees. We focus on systems whose vector fields are expressible as combinations of known basis functions. To achieve this, our approach relies on an \emph{indirect} data-driven framework, wherein offline collected data are first utilized to identify the underlying system dynamics, upon which the joint synthesis and formal verification of neural controllers and neural Lyapunov functions are subsequently executed. Building upon the paradigms in \cite{chang2019neural, zhou2022neural}, our framework extends these results across several key dimensions.

	First, we establish a structured optimization scheme and network architecture that guarantee exact theoretical soundness while enhancing numerical tractability. Specifically, a linear state normalization scheme maps the bounded domain of interest onto a well-conditioned unit hypercube. In addition, we introduce a hard bias-elimination network architecture that strictly enforces exact zero-equilibrium vanishing by construction, complemented by a pre-training stage to establish a well-conditioned initial Lyapunov topology. Building upon this architecture, Lyapunov stability conditions are formulated into a sample-based loss function, driven by a counterexample search paradigm that accelerates training efficiency by uncovering worst-case states that maximize loss violations. Remarkably, to overcome SMT solver conservatism near the equilibrium, we partition the domain into an outer SMT domain and an inner local exclusion ball, certifying the former via interval-based SMT solvers and the latter analytically via system linearization and SDP optimization. 
	Next, to enforce input constraints by construction, we integrate a hard-saturation architecture into the neural feedback design. By combining a $\mathrm{tanh}$ activation function at the controller output with linear input scaling matrices, actuator limits are strictly guaranteed. 
	Finally, to account for additive perturbations in offline identification data, robust Lyapunov synthesis conditions are derived across the entire set of data-consistent models. Crucially, rather than treating the uncertainty trade-off parameter as a pre-selected hyperparameter, our framework parameterizes this decision variable via an exponential mapping. This strictly enforces its positivity constraint by construction and enables its co-optimization along with neural network weights during backpropagation.

	\textit{Outline.} The remainder of this article is organized as follows. Section~\ref{sec:problem} formulates the system setup and problem statement. Section~\ref{sec:design} presents the identification scheme and neural controller formulation. Section~\ref{sec:synthesis_smt} details the synthesis and verification pipeline, including state normalization, network pre-training, loss function design, counterexample generation, and formal stability verification. Section~\ref{sec:input_constraints} extends this framework to handle input constraints, while Section~\ref{sec:robustness} develops the robust synthesis formulation under data perturbations. Finally, Section~\ref{sec:conclusion} concludes the paper.

	\textit{Notation.} Throughout this paper, $\mathbb{R}_{>0}^n$ denotes the set of $n$-dimensional real vectors with positive entries, and $\mathbb{S}^n$ denotes the set of $n \times n$ real symmetric matrices. A symmetric matrix $M$ is positive definite (semidefinite) if $M \succ 0$ ($M \succeq 0$), and negative definite (semidefinite) if $M \prec 0$ ($M \preceq 0$). In symmetric block matrices, the asterisk $*$ represents the symmetric transpose of the corresponding off-diagonal blocks. The notation $\operatorname{diag}(v_1, \dots, v_n)$ represents a diagonal matrix with entries $v_1, \dots, v_n$ on its main diagonal. $\|v\|_2$ and $\|v\|_\infty$ denote the 2-norm and the infinity norm of a vector $v$, respectively, while $\|M\|_2$ denotes the induced spectral norm of a matrix $M$. $\lambda_{\min}(M)$ and $\lambda_{\max}(M)$ denote the smallest and largest eigenvalues of $M$, respectively. $\nabla^2 h(x)$ denotes the Hessian matrix of a scalar function $h$. For a set $\mathcal{S} \subset \mathbb{R}^n$, $\partial \mathcal{S}$ represents its boundary, and $\Pi_{\mathcal{S}}(v)$ denotes the projection of a vector $v$ onto $\mathcal{S}$. The symbol $\vee$ denotes the logical OR operator. For real numbers $a, b \in \mathbb{R}$, $\min(a, b)$ and $\max(a, b)$ denote the minimum and maximum of $a$ and $b$, respectively. Finally, the scalar functions $\operatorname{tanh}(\cdot)$, $\operatorname{ReLU}(\cdot)$, and $\operatorname{sign}(\cdot)$ are defined as $\operatorname{tanh}(x) := \frac{e^x - e^{-x}}{e^x + e^{-x}}$, $\operatorname{ReLU}(x) := \max(0, x)$, and $\operatorname{sign}(x)$ returns $1$, $-1$, and $0$ for $x > 0$, $x < 0$, and $x = 0$, respectively, with all functions operating element-wise on vectors.

	\section{Set-up and Problem Formulation}\label{sec:problem}

	Consider a continuous-time nonlinear system of the form
	\begin{equation}
		\label{sys.d}
		\dot x = f(x,u)
	\end{equation}
	where $x\in \mathbb{R}^n$ is the state and $u\in \mathbb{R}^m$ is the control input, and $f: \mathbb{R}^n \times \mathbb{R}^m \to \mathbb{R}^n$ is an unknown vector field. Without loss of generality, we assume that the origin $(x_e,u_e)=(0,0)$ is a known equilibrium point of the system. The objective is to design a state feedback controller that stabilizes the dynamics around the origin.

	To facilitate the control design, we formalize the prior knowledge regarding the vector field $f$ through the following standing assumption.
	
	\begin{assumption} \label{ass:f}
		We know a vector-valued function $Z: \mathbb R^{n+m} \rightarrow \mathbb R^s$ such that $f(x,u) = AZ(x,u)$ for some matrix $A \in \mathbb R^{n \times s}$.
	\end{assumption}

	Under Assumption \ref{ass:f}, system \eqref{sys.d} can be equivalently rewritten as
	\begin{equation}
		\label{sys.d2}
		\dot x = A Z(x,u)
	\end{equation} 
	with $A$ unknown. Note that formulating $f(x,u)$ via a dictionary of functions as $AZ(x,u)$ is a standard paradigm in nonlinear system identification \cite{brunton2016discovering}. To account for potential inaccuracies in our prior knowledge, the function library $Z$ is permitted to include redundant terms that may not actually appear in $f(x,u)$.

	\noindent \emph{Problem.}~We are interested in the design of a neural feedback controller of the form
	\begin{equation}
		\label{eq:control}
		u=\phi (x)
	\end{equation} 
	where $\phi: \mathbb R^n \rightarrow \mathbb R^m$ is a feedforward neural network, such that the closed-loop system 
	\begin{equation} \label{eq:closed}
		\dot x = A Z(x,\phi (x))
	\end{equation}
	is asymptotically stable around the origin.

	\section{Control Design}\label{sec:design}

	To address the controller synthesis problem, we consider an \emph{indirect} method consisting of two primary steps. First, we identify the system model from collected offline data. Second, building upon the identified model, we synthesize the neural controller and its corresponding neural Lyapunov function.
	
	\subsection{System identification}
	
	To obtain an identified model of the system, we conduct an experiment on the system \eqref{sys.d} to collect the dataset
	\begin{equation} \label{dataset}
		\mathbb D := \left\{ (x(t_i),u(t_i), \dot x(t_i)) \right\}_{i=0}^{T-1}
	\end{equation} 
	where $T>0$ is the number of samples and $0\le t_0< t_1<\ldots< t_{T-1}$ are the sampling times. The samples satisfy $\dot x(t_i)=AZ(x(t_i),u(t_i))$ for $i=0,\ldots,T-1$. $Z(x(t_i),u(t_i))$ can be readily evaluated at $(x(t_i),u(t_i))$ since $Z$ is known. We organize the dataset into the following matrices:
	\begin{subequations}\label{eq:data}
		\begin{align}
			X_1 :=& \left[ \begin{matrix} \dot x(t_0) & \dot x(t_1) & \cdots & \dot x(t_{T-1}) \end{matrix} \right] 
			\in \mathbb R^{n \times T} \,, \label{eq:data1} \\
			Z_0 :=& \left[ \begin{matrix} Z(x(t_0),u(t_0)) & Z(x(t_1),u(t_1)) \end{matrix} \right.\nonumber
			\\&\quad \left. \begin{matrix} \cdots & Z(x(t_{T-1}),u(t_{T-1})) \end{matrix} \right]
			\in \mathbb R^{s \times T}, \label{eq:data2}
		\end{align} 
	\end{subequations}
	and these data matrices satisfy
	\begin{equation} \label{dataiden_free}
		X_1=A Z_0.
	\end{equation} 
	For identification purpose, we assume that matrix $Z_0$ has full row rank. This condition necessitates $T \ge s$ and can be interpreted as a data \emph{richness} requirement, which can typically be fulfilled by acquiring additional samples.
	
	\begin{assumption}\label{ass:Z0} 
		$Z_0$ has full row rank.
	\end{assumption}
	
	By \eqref{dataiden_free}, Assumption \ref{ass:Z0} implies that
	\begin{equation}\label{eq:A}
		A= X_1 Z_0^\dag,
	\end{equation} 
	where $Z_0^\dag \in \mathbb R^{T \times s}$ denotes the right inverse of $Z_0$. Hence, the identified model can be expressed as
	\begin{equation}\label{eq:iden}
		\dot x = X_1 Z_0^\dag Z(x,u).
	\end{equation} 
	
	\subsection{Neural controller design}
	
	The first statement below presents the main result of this section, which establishes a sufficient condition for asymptotic stabilization through the joint design of a neural feedback controller $u=\phi (x)$ and a neural Lyapunov function $V(x)$ with $V: \mathbb R^n \rightarrow \mathbb R$ a feedforward neural network.

	\begin{definition} \label{def:roa}
		Let $x_e$ be an asymptotically stable equilibrium point for the system $\dot{x} = f(x)$. A set $\mathcal{R} \subseteq \mathbb{R}^n$ is called a region of attraction (ROA) for the system relative to $x_e$ if, for every initial state $x(0) \in \mathcal{R}$, the corresponding state trajectory $x(t)$ satisfies $\lim_{t \to \infty} x(t) = x_e.$				
	\end{definition}

	\begin{theorem} \label{thm:control.free}
		Consider the nonlinear system \eqref{sys.d}. Let Assumptions \ref{ass:f} and \ref{ass:Z0} hold. 
		Let $\mathcal{X}\subseteq \mathbb{R}^n$ be a set containing the origin. 
		Suppose there exist neural networks $\phi$ and $V$ satisfying $\phi(0)=0$ and $V(0)=0$ such that for each $x \in \mathcal{X}\setminus \{0\}$
		\begin{subequations}\label{stability.con.free}
			\begin{align}
				&V(x)>0, \label{stability.con1.free}\\
				&\frac{\partial V(x)}{\partial x} X_1 Z_0^\dag Z(x,\phi(x)) <0. \label{stability.con2.free}
			\end{align} 
		\end{subequations}
		Then, the origin is asymptotically stable for the closed-loop system \eqref{eq:closed} and any sublevel set $\mathcal{V}:=\left\{x\in \mathbb{R}^n\colon V(x) \le \delta\right\}$, with $\delta> 0$ such that $\mathcal{V}$ is contained in $\mathcal{X}$ defines an estimate of the ROA relative to $x=0$.
	\end{theorem}

	\begin{pf}
		The conditions \eqref{stability.con.free} imply that $V$ serves as a valid Lyapunov function for the closed-loop system \eqref{eq:closed}, which completes the proof. \hfill$\square$
	\end{pf}

	\subsection{Comparison With Other Methods} 
	
	To highlight the methodological differences, it is useful to compare the proposed approach with existing data-driven methods developed for the nonlinear system \eqref{sys.d}, such as~\cite{dprt2023cancellation, hu2025enforcing, guo2024data}.

	First, consider the approximate nonlinearity cancellation method in~\cite{dprt2023cancellation} and the data-driven contraction framework in~\cite{hu2025enforcing}. By introducing the augmented state-input coordinate $\xi :=[\begin{smallmatrix} x \\ u \end{smallmatrix}]$, system \eqref{sys.d} can be recast into $\dot x = \overline A \mathcal{Z}(\xi) $. Here, $\overline A \in \mathbb R^{n \times R}$ and $\mathcal{Z}(\xi)= \begin{bmatrix} \xi \\ \mathcal{Q}(\xi) \end{bmatrix}$ constitutes a rearrangement of the original dictionary $Z(x,u)$, thereby achieving a decoupling between the linear terms $\xi$ and the remaining nonlinearities captured by $\mathcal{Q}(\xi)$. To further yield a transformed system with a constant input matrix, an integral control action $\dot u = v$ is implemented, where $v \in \mathbb{R}^m$ serves as the new control input. Consequently, by partitioning the matrix as $\overline A := \begin{bmatrix} \overline A_1 & \overline A_2 \end{bmatrix}$ with $\overline A_1 \in \mathbb R^{n \times (n+m)}$ and $\overline A_2 \in \mathbb R^{n \times (R-n-m)}$, the overall augmented system can be compactly expressed as $\dot \xi = \mathcal{A} \mathcal{Z}(\xi)+\mathcal{B} v$, where $\mathcal{A} := \begin{bmatrix} \, \overline A_1 & \overline A_2 \\ {0}_{m\times (n+m)} & {0}_{m\times (R-n-m)} \\ \end{bmatrix}$ and $\mathcal{B} := \begin{bmatrix} 0_{n\times m} \\ I_m \end{bmatrix}$. The controller to be synthesized takes the form $v = \mathcal{K} \mathcal{Z}(\xi)$, where $\mathcal{K} = [\,\mathcal{K}_1 \quad \mathcal{K}_2\,]$ with $\mathcal{K}_1 \in \mathbb{R}^{m \times (n+m)}$ and $\mathcal{K}_2 \in \mathbb{R}^{m \times (R-n-m)}$.

	In \cite{dprt2023cancellation}, controller synthesis relies on a nonlinearity cancelation framework aimed at stabilizing the linear part while minimizing the impact of the nonlinear part in closed-loop. This guarantees asymptotic stability of the origin provided that the underlying SDP is feasible and the nonlinearities satisfy $\lim_{\xi \rightarrow 0} \frac{\|\mathcal{Q}(\xi)\|_2}{\|\xi\|_2}=0$. On the other hand, in \cite{hu2025enforcing}, the method focuses on enforcing exponential contractivity on the closed-loop dynamics. Compared with our framework, a fundamental limitation shared by \cite{dprt2023cancellation} and \cite{hu2025enforcing} is that their feasibility necessitates the stabilizability of the structured pair $(\begin{bmatrix} \, \overline A_1 \\ 0 \end{bmatrix} , \mathcal{B} )$. In contrast, our neural controller is synthesized by directly taking the entire nonlinear dynamics into account, thereby bypassing such linear-structure stabilizability requirements.

	A further distinction lies in the capacity to achieve \emph{global} stability guarantees. Specifically, the framework in \cite{hu2025enforcing} can theoretically establish global stability provided that the contractivity condition holds over $\mathcal{X}=\mathbb{R}^n$. On the other hand, global stability is structurally unachievable within the framework of \cite{dprt2023cancellation} because the specific zero-block configuration of the input matrix $\mathcal{B}$ renders the exact nonlinearity cancellation, namely, $\begin{bmatrix} \, \overline A_2 \\ 0 \end{bmatrix} + \mathcal{B} \mathcal{K}_2 =0$, impossible to satisfy. 
	From a formal verification perspective, although algebraic SMT solvers, such as Z3, can analyze unbounded spaces, their theoretical guarantees are confined to polynomial systems \cite{edwards2025general}. In fact, for dynamics involving transcendental nonlinearities, such as the system's $\sin(\cdot)$ terms or the neural network's $\tanh(\cdot)$ activations, global verification is mathematically undecidable \cite{gao2012delta}. To accommodate these transcendental functions, we utilize interval-based SMT solvers. Crucially, such solvers rely on interval constraint propagation over compact domains, and hence restricting the verification space to a bounded domain $\mathcal{X}$ is an essential requirement for returning formal stability certificates.

	Another relevant line of research is the Taylor's expansion-based data-driven framework in \cite{guo2024data}. An appealing feature of this approach is that it does not require prior knowledge of the underlying function library. Instead, it approximates the unknown dynamics \eqref{sys.d} around the equilibrium via Taylor's polynomials, and treats the truncated remainder as a bounded uncertainty. Nonetheless, to maintain computational tractability within SOS and SDP formulations, the candidate controllers and Lyapunov functions in~\cite{guo2024data} are restricted to linear or fixed-degree polynomial templates. Furthermore, the truncated remainder grows unbounded as $\Vert{}x\Vert{}_2 \to \infty$, which inherently confines the theoretical guarantees to a local region.

	\section{Neural Controller Synthesis and SMT-Based Lyapunov Verification} \label{sec:synthesis_smt}
	This section develops the data-driven synthesis and formal verification framework for the neural controller and neural Lyapunov function introduced in Theorem \ref{thm:control.free}. The overall scheme relies on minimizing a sample-based loss function to enforce stability conditions, while employing SMT solvers along with local stability analysis to formally certify asymptotic stability.
	
	\subsection{Neural Synthesis Framework and Optimization Scheme}
	\label{subsec:neural_synthesis_optimization}
	\subsubsection{Dynamic Data Generation and State Normalization} 
	Rather than relying on a static offline dataset, a dynamic data-augmentation strategy is adopted to train both the neural controller and the neural Lyapunov function. Specifically, the training set $\overline{\mathbb{D}}$ is initially populated via uniform state sampling from the domain of interest. Subsequently, an adaptive, counterexample-guided mechanism is employed to iteratively expand the dataset with critical samples, as detailed later in Sections \ref{subsec:counterexample_generation} and \ref{subsec:stability_formal_verification}.

	In multivariable nonlinear systems, the states typically operate on disparate physical scales, and directly injecting raw state data into neural networks would induce numerical ill-conditioning and sluggish gradient propagation, especially given that standard activation functions, such as $\tanh(\cdot)$, suffer from vanishing gradients outside the interval $[-1, 1]$. Hence, state normalization serves not merely as an empirical data-preprocessing step, but a fundamental requirement for stable and effective network optimization.

	To this end, we introduce a linear coordinate transformation for state normalization. Consider the bounded domain $\mathcal{X} := \{x \in \mathbb{R}^n \mid -b_i \le x_i \le b_i, \; i = 1, \dots, n\}$, where $b := [b_1, \dots, b_n]^\top \in \mathbb{R}^n_{>0}$ denotes the vector of coordinate-wise upper bounds. By scaling the raw state $x$ relative to $b$, the normalized state vector is defined as $z := [z_1, \dots, z_n]^\top$ with $z_i = x_i / b_i, i = 1, \dots, n$. This transformation maps the domain $\mathcal{X}$ onto a well-conditioned unit hypercube $[-1, 1]^n$, and $z$ are subsequently utilized for neural network training. Let $S := \text{diag}(\frac{1}{b_1}, \frac{1}{b_2}, \dots, \frac{1}{b_n})$ and hence $z = Sx$. Consequently, the transformed dynamics governing $z$ are formulated as:
	\begin{equation}
		\label{sys.d3}
		\dot z = M \Phi(z, u)
	\end{equation} 
	where $\Phi(z,u) := T_Z Z(x,u)$ represents the normalized function library, $T_Z \in \mathbb{R}^{s \times s}$ is a non-singular diagonal matrix, and $M \in \mathbb{R}^{n \times s}$ denotes the system matrix to be identified. The normalized data matrices are organized from the dataset $\mathbb{D}$ as:
	\begin{subequations} \label{eq:normalized_data_matrices}
		\begin{align}
			Z_1 :=& \begin{bmatrix} S\dot{x}(t_0) & S \dot{x}(t_1) & \cdots & S \dot{x}(t_{T-1}) \end{bmatrix} \in \mathbb R^{n \times T}, \\
			\Phi_0 :=& \left[ \begin{matrix} \Phi(S x(t_0),u(t_0)) & \Phi(S x(t_1),u(t_1)) \end{matrix} \right.\nonumber
			\\&\quad \left. \begin{matrix} \cdots & \Phi(S x(t_{T-1}),u(t_{T-1})) \end{matrix}  \right] \in \mathbb R^{s \times T}.
		\end{align}
	\end{subequations}
	These matrices satisfy $Z_1 = M\Phi_0$, and $\Phi_0$ preserves the full row rank property of $Z_0$ since $T_Z$ is non-singular. Consequently, $M$ is identified as $M = Z_1 \Phi_0^\dagger$,
	where $\Phi_0^\dagger \in \mathbb{R}^{T \times s}$ denotes the right inverse of $\Phi_0$. The identified normalized model is thus expressed as
	\begin{equation} \label{eq:normalized_model_identified}
		\dot{z} = Z_1 \Phi_0^\dagger \Phi(z,u).
	\end{equation}
	As formalized below, the stability of the normalized system \eqref{eq:normalized_model_identified} is equivalent to that of the original physical dynamics \eqref{eq:iden}.
	
	\begin{proposition} \label{prop:equivalence}
		Consider the system \eqref{eq:iden} and the transformed normalized system \eqref{eq:normalized_model_identified}, subject to the state transformation $z=Sx$. Let $\mathcal{X}_z := [-1, 1]^n$ denote the unit hypercube, and let $\mathcal{X} \subseteq \mathbb{R}^n$ be the domain containing the origin such that $\mathcal{X}_z = S\mathcal{X}$. Suppose there exist $\psi$ and $V_z$ satisfying $\psi(0)=0$ and $V_z(0)=0$. Then, the stability conditions
		\begin{subequations}\label{stability.con.normalized.prop}
			\begin{align}
				&V_z(z)>0, \label{stability.prop.norm1}\\
				&\frac{\partial V_z(z)}{\partial z} Z_{1} 	\Phi_{0}^{\dagger} \Phi(z,\psi(z)) < 0, \label{stability.prop.norm2}
			\end{align}
		\end{subequations}
		hold for all $z\in\mathcal{X}_z \backslash\{0\}$ if and only if 
		\begin{subequations}\label{stability.con.physical.prop}
			\begin{align}
				&V_z(Sx)>0, \label{stability.prop.phys1}\\
				&\frac{\partial V_z(Sx)}{\partial x} X_{1} 	Z_{0}^{\dagger} Z(x,\psi(Sx)) < 0, \label{stability.prop.phys2}
			\end{align}
		\end{subequations}
		hold for all $x\in \mathcal{X} \backslash\{0\}$.
	\end{proposition}

	\begin{pf}
		Since $S$ is a non-singular diagonal matrix, $z=Sx$ defines a diffeomorphism between $\mathcal{X}$ and $\mathcal{X}_z$, which establishes a bijective mapping such that $x=0 \iff z=0$ and $\mathcal{X}\backslash\{0\} \iff \mathcal{X}_z\backslash\{0\}$. For any state pair adhering to this transformation, $V_z(Sx) = V_z(z)$ holds, which implies that \eqref{stability.prop.norm1} holds for all $z \in \mathcal{X}_z\setminus\{0\}$ if and only if \eqref{stability.prop.phys1} holds for all $x \in \mathcal{X}\setminus\{0\}$.

		Furthermore, since $\dot{z} = S\dot{x}$, \eqref{eq:iden} yields $\dot{z} = S X_1 Z_0^\dagger Z(x,u)$. By \eqref{eq:normalized_model_identified}, we have $S X_1 Z_0^\dagger Z(x, u) = Z_1 \Phi_0^\dagger \Phi(z, u)$. Under the controller $u = \psi(z) = \psi(Sx)$, it yields
		\[
		S X_1 Z_0^\dagger Z(x, \psi(Sx)) = Z_1 \Phi_0^\dagger \Phi(z, \psi(z)),
		\]
		and it follows that
		\begin{equation*}
			\begin{aligned}
				\frac{\partial V_z(Sx)}{\partial x} X_{1} Z_{0}^{\dagger} Z(x,\psi(Sx)) 
				&= \frac{\partial V_z(z)}{\partial z} S X_1 	Z_0^\dagger Z(x, \psi(Sx)) \\
				& = \frac{\partial V_z(z)}{\partial z} Z_1 	\Phi_0^\dagger \Phi(z, \psi(z)).
			\end{aligned}
		\end{equation*}
		Consequently, \eqref{stability.prop.norm2} holds for all $z \in \mathcal{X}_z \setminus \{0\}$ if and only if \eqref{stability.prop.phys2} holds for all $x \in \mathcal{X} \setminus \{0\}$, which completes the proof. \hfill$\square$
	\end{pf}

	Based on the equivalence established in Proposition \ref{prop:equivalence}, the synthesis objective reduces to designing a neural controller $u=\psi(z)$ and a neural Lyapunov function $V_z(z)$ that satisfy \eqref{stability.con.normalized.prop} over $\mathcal{X}_z \setminus \{0\}$. Consequently, the stability conditions \eqref{stability.con.free} are guaranteed to hold across $\mathcal{X} \setminus \{0\}$ by recovering the controller as $\phi(x)=\psi(Sx)$ and the Lyapunov function as $V(x)=V_z(Sx)$.

	\begin{rem} \label{rem:asymmetric_bounds}
		For asymmetric state domains characterized by $\mathcal{X} = \{x \in \mathbb{R}^n \mid -a_i \le x_i \le b_i, \; i = 1, \dots, n\}$ with $a_i \neq b_i$, state normalization is achieved via the maximum-absolute-value scaling $s_i := \max(a_i, b_i)$. Consequently, the resulting coordinate transformation $z_i = x_i / s_i$, $i = 1, \dots, n$ maps $\mathcal{X}$ onto a hyper-rectangle embedded within $[-1, 1]^n$, while remaining linear and origin-preserving.  \hfill $\blacksquare$
	\end{rem}

	\subsubsection{Structural Equilibrium Constraints and Network Pre-training}
	The theoretical prerequisites of the Lyapunov stability conditions in Theorem \ref{thm:control.free} and Proposition \ref{prop:equivalence} necessitate that both the Lyapunov function and the controller vanish at the equilibrium, i.e., $V_z(0) = 0$ and $\psi (0) = 0$. Rather than imposing these equilibrium constraints via soft regularization terms within the loss function, which cannot guarantee exact zero-convergence, we adopt a hard bias-elimination network architecture formulated as:
	\begin{subequations}\label{equilibrium.constraint}
		\begin{align}
			V_z(z) &= \mathcal{N}_V(z) - \mathcal{N}_V(0), \label{eq:V_zero_bias} \\
			\psi(z) &= \mathcal{N}_u(z) - \mathcal{N}_u(0), \label{eq:u_zero_bias}
		\end{align} 
	\end{subequations}
	where $\mathcal{N}_V: \mathbb{R}^n \to \mathbb{R}$ and $\mathcal{N}_u: \mathbb{R}^n \to \mathbb{R}^m$ represent the raw outputs of the underlying multilayer perceptrons (MLPs). By explicitly subtracting the network evaluations at the origin, $V_z(0) = 0$ and $\psi(0) = 0$ are strictly guaranteed by construction, which is independent of the network weights and biases.

	During the joint synthesis, optimizing both networks simultaneously from random initializations might suffers from numerical instability or convergence failure. To overcome this challenge, we execute a pre-training stage wherein $V_z(z)$ is first trained to approximate a canonical quadratic function:
	\[
	V_{\text{target}}(z) = \beta\|z\|_2^2,
	\]
	where $\beta > 0$ is a positive scalar. This pre-training step establishes a well-conditioned initial Lyapunov topology that guides the subsequent co-optimization of $\psi(z)$ and $V_z(z)$.
	
	\subsubsection{Loss Function Design and Optimization}
	The loss function architecture is designed to explicitly penalize violations of the Lyapunov stability conditions \eqref{stability.con.normalized.prop}. To transform \eqref{stability.con.normalized.prop} into penalty terms within loss function, a positive-definite metric is defined as $\Omega(z) := z^\top z$.

	Given a training set $\overline{\mathbb{D}}$ consisting of $N$ normalized state samples and a positive margin parameter $\varepsilon > 0$, the penalty term enforcing the positive definiteness of $V_z(z)$ over $\mathcal{X}_z \setminus \{0\}$ is formulated as:
	\begin{equation}\label{loss.V}
		L_{\text{PD}} (V_z) = \frac{1}{N} \sum_{k=1}^{N} \text{ReLU}\left(\varepsilon \cdot \Omega(z_k) - V_z(z_k)\right),
	\end{equation}
	and the penalty term enforcing the negative definiteness of the time derivative $\dot{V}_z(z)$ over $\mathcal{X}_z \setminus \{0\}$ is given by:
	\begin{equation}\label{loss.nd}
		\begin{aligned}
			L_{\text{ND}}(V_z, \psi) =  \frac{1}{N} \sum_{k=1}^{N} 	\text{ReLU} & \left( \frac{\partial V_z(z_k)}{\partial z}  Z_{1} \Phi_{0}^{\dagger} \Phi(z_k,\psi(z_k)) \right. \\
			& \quad \left. + \varepsilon \cdot \Omega(z_k) \right).
		\end{aligned}
	\end{equation}
	Then, the overall loss function for jointly synthesizing the neural controller $\psi(z)$ and the neural Lyapunov function $V_z(z)$ is defined as:
	\begin{equation}\label{loss.total}
		L_{\text{total}}(V_z, \psi) = \gamma_1 L_{\text{PD}}(V_z) + \gamma_2 L_{\text{ND}}(V_z, \psi),
	\end{equation}
	where $\gamma_1 > 0$ and $\gamma_2 > 0$ are positive penalty hyperparameters.

	\begin{rem} \label{rem:epsilon_selection_logic}
		The parameter $\varepsilon$ governs the closed-loop decay rate, and its selection embodies a trade-off between the transient performance and optimization tractability. Specifically, a large margin $\varepsilon$ imposes a restrictive constraint on \eqref{stability.prop.norm2}, which forces the controller $\psi(z)$ to deliver high-gain control actions, thereby potentially inducing gradient explosion and optimization divergence. On the other hand, selecting a conservative margin $\varepsilon$ dampens the closed-loop convergence rate, which results in a sluggish transient response. Furthermore, from a verification perspective, an overly small $\varepsilon$ severely compromises the numerical robustness of interval-based SMT solvers. When $\varepsilon$ falls below the solver's precision tolerance, it might induce computational bottlenecks. Consequently, $\varepsilon$ should be chosen to reconcile closed-loop performance with the solver's numerical limits.  \hfill $\blacksquare$
	\end{rem}

	\begin{rem} \label{rem:loss_weights_scaling}
		The selection of $\gamma_1$ and $\gamma_2$ in \eqref{loss.total} is intrinsically determined by the normalization scaling bounds. For $\mathcal{X}$ characterized by small scaling bounds, i.e., $0<b_i \ll 1$, $i = 1, \dots, n$, the coordinate transformation substantially amplifies $\dot{z} = S \dot{x}$. Consequently, $L_{\text{ND}}$ yields disproportionately large gradients that suppress the contribution of $L_{\text{PD}}$ and dominate the loss function. In this case, a higher weight ratio $\gamma_1 / \gamma_2$ is required to balance the respective contributions of $L_{\text{ND}}$ and $L_{\text{PD}}$. Conversely, for $\mathcal{X}$ characterized by large scaling bounds, i.e., $b_i \gg 1$, $i = 1, \dots, n$, coordinate transformation dampens the $z$ dynamics, and thus a lower weight ratio $\gamma_1 / \gamma_2$ is essential to prevent $L_{\text{ND}}$ from being numerically neglected during optimization. \hfill $\blacksquare$
	\end{rem}
	\subsection{Counterexample Generation} \label{subsec:counterexample_generation}
	Conventional deep learning paradigms predominantly rely on Empirical Risk Minimization (ERM), which optimizes average performance over randomly drawn samples \cite{vapnik1999overview, goodfellow2016deep}. However, certifying Lyapunov stability demands a deterministic, worst-case guarantee across the entire domain $\mathcal{X}_z$. Instead of blindly sampling from $\mathcal{X}_z$, we shift the focus toward a targeted search for counterexamples. A counterexample $z_{ce}$ is defined as a state vector within $\mathcal{X}_z \setminus \{0\}$ that violates the stability conditions \eqref{stability.con.normalized.prop}. Mathematically, the violation set $\mathcal{X}_{ce}$, which represents the union of the individual violation zones, is formalized as:
	\[
	\mathcal{X}_{ce} = \left\{ z \in \mathcal{X}_z \setminus \{0\} \mid (V_z(z) < \varepsilon \Omega(z)) \lor (\dot{V}_z(z) > -\varepsilon \Omega(z)) \right\}.
	\]	
	Here, we introduce three computationally efficient approaches for counterexample generation: Projected Gradient Descent (PGD), Generative Adversarial Network (GAN) and Particle Swarm Optimization (PSO).

	First, the PGD algorithm is employed as a \emph{gradient-based} search mechanism to uncover the worst-case states that maximize the violation loss \cite{wu2023neural}. Specifically, the search is initialized by generating a seed dataset $\mathbb{D}_{\text{pgd}}$ containing $N_{\text{pgd}}$ samples within $\mathcal{X}_z$. Over an optimization horizon of $p$ steps, each state trajectory in $\mathbb{D}_{\text{pgd}}$ is recursively driven along the worst-case gradient direction with a step size $\kappa$. The updated vector is then projected back onto $\mathcal{X}_z$ via the projection operator $\Pi_{\mathcal{X}_z}$, yielding the following update law for $t = 0, \dots, p-1$:
	\[
	z^{t+1} = \Pi_{\mathcal{X}_z} \left( z^t + \kappa \cdot \text{sign}\left( \left. \frac{\partial L_{\text{total}}(V_z, \psi)}{\partial z} \right|_{z = z^t} \right) \right).
	\]
	If $z^t \in \mathcal{X}_{ce}$, $t = 0, \dots, p-1$, it is classified as a counterexample.

	Second, the GAN algorithm is introduced as a \emph{mapping-based} counterexample generator \cite{goodfellow2020generative}. Unlike PGD, which recursively updates state trajectories step-by-step using gradient information, this generative mechanism establishes a direct map from a continuous stochastic input space to candidate counterexamples. Formally, the generator $G_\mu$, parameterized by trainable weights $\mu$, takes a random vector $\tau$ drawn from a standard Gaussian distribution $p_\tau$ as its input, and maps it directly to a candidate counterexample:
	\[
	z_{\text{gan}} = G_\mu(\tau).
	\]
	In conventional GANs, the discriminator is an empirical neural network, which serves as the optimization target for the generator $G_\mu$. In contrast, our framework replaces the discriminator with the deterministic loss function $L_{\text{total}}$. The objective is thus to maximize $L_{\text{total}}$ to uncover counterexamples. If $z_{\text{gan}} \in \mathcal{X}_{ce}$, it is classified as a counterexample.

	Finally, the PSO algorithm is utilized as a \emph{derivative-free} search mechanism to uncover the worst-case states that maximize the violation loss \cite{kennedy1995particle}, thereby bypassing the gradient vanishing or explosion issues caused by saturating activation functions (e.g., $\tanh$) and non-smooth penalty terms (e.g., $\text{ReLU}$). Specifically, the search is initialized by deploying a swarm $\mathbb{D}_{\text{pso}}$ consisting of $N_{\mathrm{pso}}$ particles within $\mathcal{X}_z$, where each particle $i \in \{1, \dots, N_{\mathrm{pso}}\}$ denotes a state vector $z^{(i), t} \in \mathcal{X}_z$ associated with a search velocity vector $v^{(i), t} \in \mathbb{R}^n$. Over an optimization horizon of $p$ steps, the trajectory of each particle is recursively updated in a derivative-free manner driven by individual and collective historical optima rather than gradient evaluations, yielding the following iterative update law for $t = 0, \dots, p-1$:
	\begin{align*}
		v^{(i), t+1} &= \omega v^{(i), t} + c_1 r_1 \left( q^{(i), t} - z^{(i), t} \right) \\
		&\quad + c_2 r_2 \left( g^t - z^{(i), t} \right), \\
		z^{(i), t+1} &= \Pi_{\mathcal{X}_z} \left( z^{(i), t} + v^{(i), t+1} \right),
	\end{align*}
	where $\omega$ denotes the inertia weight, $c_1$ and $c_2$ represent individual and collective tracking gains, respectively, and $r_1, r_2 \sim \mathcal{U}(0,1)$ are independent random variables introducing stochastic exploration. The vectors $q^{(i), t}$ and $g^t$ denote the individual historical best position and the global best position achieved across $\mathbb{D}_{\text{pso}}$ up to step $t$, respectively, which are extracted by maximizing 
	\begin{align*}
		q^{(i), t} &= \arg\max_{z \in \{z^{(i), 0}, \dots, z^{(i), t}\}} L_{\text{total}}(V_z, \psi)\big|_z, \\
		g^t &= \arg\max_{q \in \{q^{(1), t}, \dots, q^{(N_{\mathrm{pso}}), t}\}} L_{\text{total}}(V_z, \psi)\big|_q.
	\end{align*}
	If $z^{(i), t} \in \mathcal{X}_{ce}$ for $i = 1, \dots, N_{\mathrm{pso}}$ and $t = 0, \dots, p-1$, it is classified as a counterexample.

	Throughout the counterexample search phase, the parameters of both $V_z$ and $\psi$ remain fixed. The underlying search paradigm features a highly flexible architecture, capable of deploying PGD, GAN, and PSO either individually, in pairwise combinations, or through a tripartite integration. Ultimately, the uncovered counterexamples are appended into $\overline{\mathbb{D}}$ to train $V_z$ and $\psi$.						

	\subsection{Stability Formal Verification} \label{subsec:stability_formal_verification}
	Neural network training inherently operates on finite, discrete samples. Even when augmented with adversarial and heuristic search paradigms such as PGD, GAN and PSO, data-driven optimization can only ensure that stability conditions hold at discrete sample points. In contrast, certifying Lyapunov stability demands a continuous guarantee that holds universally across the domain $\mathcal{X}_z$. To bridge this methodological gap, interval-based SMT solvers, such as dReal \cite{gao2013dreal}, are deployed in this paper to execute formal stability verification. To render the verification process computationally tractable, the solver operates under the paradigm of $\delta$-completeness \cite{gao2012delta}, wherein the target domain $\mathcal{X}_z$ is partitioned via recursive state-space bisection governed by a precision tolerance $\delta > 0$. In fact, the mathematical foundation of the formal verification procedure rests upon interval arithmetic and interval constraint propagation \cite{benhamou2006continuous}. Specifically, for each generated sub-domain, the solver over-approximates the sub-domain via an interval vector $[z]$. According to the fundamental theorem of interval arithmetic, the solver evaluates the functions on the left-hand side of \eqref{stability.con.normalized.prop} over $[z]$, yielding a conservative interval enclosure that strictly contains the true image of these functions over that sub-domain. If this interval enclosure satisfies the stability conditions \eqref{stability.con.normalized.prop}, the associated sub-domain is successfully verified and excluded from further search. Otherwise, the sub-domain is marked as unverified and subjected to further subdivision. For any sub-domain, the subdivision process terminates either upon verification, or when the interval width falls below $\delta$ without satisfying \eqref{stability.con.normalized.prop}, in which case the solver finds a violation and returns this falsifying sub-domain. Verification across $\mathcal{X}_z$ is successful if all partitioned sub-domains are verified. Consequently, upon successful verification, the Lyapunov conditions \eqref{stability.con.normalized.prop} are rigorously guaranteed to hold continuously across $\mathcal{X}_z$, thereby elevating empirical optimization into a deterministic stability certificate.

	Despite their mathematical rigor, interval-based SMT solvers encounter severe computational bottlenecks near the origin. Since both $V_z(z)$ and $\psi(z)$ vanish at the equilibrium $z=0$, the conservative interval enclosures of the Lyapunov stability conditions near the origin inherently contain zero, which prevents the solver from certifying \eqref{stability.con.normalized.prop}. To circumvent this verification limitation, we partition $\mathcal{X}_z$ by isolating the origin, thereby shifting the SMT verification domain from $\mathcal{X}_z \setminus \{0\}$ to $\mathcal{X}_{\text{SMT}} := \mathcal{X}_z \setminus \mathcal{B}_{\sigma}(0)$, where $\mathcal{B}_{\sigma}(0) := \{z \in \mathbb{R}^n \mid \|z\|_2 < \sigma\}$ with a small radius $\sigma > 0$. 
	Depending on the verification outcome, the training loop branches into two pathways. On the one hand, if the SMT solver uncovers a stability violation and returns a falsifying interval vector, candidate counterexamples $z_{\text{SMT}}$ are sampled around this interval. Those satisfying $z_{\text{SMT}} \in \mathcal{X}_{ce}$ are classified as counterexamples and aggregated into the training set $\overline{\mathbb{D}}$ to drive the next training loop. On the other hand, if the SMT solver successfully verifies $\mathcal{X}_{\text{SMT}}$, the remaining stability analysis within $\mathcal{B}_{\sigma}(0)$ is conducted via the local Lyapunov approach detailed below.

	Let $f_{\text{cl}}(z) := Z_{1} 	\Phi_{0}^{\dagger} \Phi(z,\psi(z))$. The first-order Taylor's expansion of the closed-loop system $\dot{z} = f_{\text{cl}}(z)$ around the origin is given by
	\begin{equation} \label{eq:closed_loop_taylor}
		\dot{z} = A_{\text{cl}} z + g(z),
	\end{equation}
	where $A_{\text{cl}} := \left. \frac{\partial f_{\text{cl}}(z)}{\partial z} \right|_{z=0}$ and $g(z)$ denotes the remainder satisfying $g(0) = 0$ and $\left. \frac{\partial g(z)}{\partial z} \right|_{z=0} = 0$. If the matrix $A_{\text{cl}}$ is not Hurwitz, the candidate controller is rejected, and the algorithm reverts to the training loop to re-synthesize $\psi(z)$ and $V_z(z)$. On the other hand, if the matrix $A_{\text{cl}}$ is Hurwitz, we proceed with the following local stability certification.

	Consider a candidate local Lyapunov function $V_{\text{loc}}(z) := z^\top P z$ with $P \in \mathbb S^{n \times n}$. Its time derivative along \eqref{eq:closed_loop_taylor} is given by:
	\begin{equation} \label{eq:v_dot_expansion}
		\dot{V}_{\text{loc}}(z) = \dot{z}^\top P z + z^\top P \dot{z} = z^\top (A_{\text{cl}}^\top P + P A_{\text{cl}}) z + 2 z^\top P g(z).
	\end{equation}
	We search for a matrix $P \succ 0$ via the following SDP:
	\begin{equation} \label{eq:lmi_sdp_opt}
		\begin{aligned}
			\max_{P \succ 0,\, \gamma > 0} \quad & \gamma + \nu \log \det(P) \\
			\text{s.t.} \quad & P \preceq I, \\
			& A_{\mathrm{cl}}^\top P + P A_{\mathrm{cl}} + \gamma I \preceq 0,
		\end{aligned}
	\end{equation}
	where $\nu > 0$ is a regularization parameter.
	Since $A_{\text{cl}}$ is Hurwitz, the SDP \eqref{eq:lmi_sdp_opt} is feasible. Let $P$ denote its solution and define $Q := -(A_{\text{cl}}^\top P + P A_{\text{cl}})$. Then, the quadratic term in \eqref{eq:v_dot_expansion} satisfies 
	\begin{equation} \label{eq:quadratic_term_bound}
		z^\top (A_{\text{cl}}^\top P + P A_{\text{cl}}) z = -z^\top Q z \le -\lambda_{\min}(Q) \|z\|_2^2.
	\end{equation}		
	Furthermore, by the Cauchy-Schwarz inequality along with Taylor's theorem, the higher-order cross-term in \eqref{eq:v_dot_expansion} satisfies
	\begin{equation} \label{eq:higher_order_cross_term_bound}
		\begin{aligned}
			2 z^\top P g(z) &= 2 \sum_{i=1}^n (P_{i,\cdot} z) g_i(z) \\
			&\le 2 \sum_{i=1}^n \|P_{i,\cdot}\|_2 \|z\|_2 \cdot \left( \frac{1}{2} H_i \|z\|_2^2 \right),
		\end{aligned}
	\end{equation}
	where $P_{i,\cdot}$ denotes the $i$-th row of $P$, and $H_i := \max_{z \in \mathcal{B}_{\sigma}(0)} \|\nabla^2 f_{\text{cl}, i}(z)\|_2$ denotes an upper bound on the Hessian spectral norm for the $i$-th component of $f_{\text{cl}}$ over $\mathcal{B}_{\sigma}(0)$. Note that the evaluation of $A_{\text{cl}}$ and $\nabla^2 f_{\text{cl}, i}(z)$ implicitly requires the function library $Z(x,u)$ and the activation functions in $\psi(z)$ to be twice continuously differentiable over $\mathcal{B}_{\sigma}(0)$.
	Defining $K_g := \sum_{i=1}^n \Vert{}P_{i,\cdot}\Vert{}_2 H_i$ and combining \eqref{eq:quadratic_term_bound} with \eqref{eq:higher_order_cross_term_bound}, we arrive at
	\begin{equation} \label{eq:v_dot_final_bound}
		\dot{V}_{\text{loc}}(z) \le -\left( \lambda_{\min}(Q) - K_g \|z\|_2 \right) \|z\|_2^2.
	\end{equation}
	Consequently, $\dot{V}_{\text{loc}}(z) < 0$ holds for all $z \in \mathcal{B}_{\sigma_{\mathrm{loc}}}(0) \setminus \{0\}$, with $\sigma_{\mathrm{loc}} := \frac{\lambda_{\min}(Q)}{K_g}$.

	Building upon the certified local stability and the outer SMT verification, we now rigorously establish the conditions under which every trajectory initialized within a sublevel set $\mathcal{R}_c := \{z \in \mathbb{R}^n \mid V_z(z) \le c\}$ with $c > 0$ contained in $\mathcal{X}_z$ converges to the origin. First, since SMT verification certifies that $\dot{V}_z(z) \le -\varepsilon \Vert{}z\Vert{}_2^2 < 0$ over $\mathcal{R}_c \setminus \mathcal{B}_\sigma(0)$, any state trajectory starting in $\mathcal{R}_c$ remains within $\mathcal{R}_c$ for all $t \ge 0$ and enters $\mathcal{B}_\sigma(0)$ in finite time. Second, let $\mathcal{R}_{\mathrm{loc}} := \{ z \in \mathbb{R}^n \mid z^\top P z \le \lambda_{\min}(P) \sigma_{\mathrm{loc}}^2 \}$. Since $z^\top P z \ge \lambda_{\min}(P) \Vert{}z\Vert{}_2^2$, we have $\Vert{}z\Vert{}_2 \le \sigma_{\mathrm{loc}}$, which implies $\mathcal{R}_{\mathrm{loc}} \subseteq \mathcal{B}_{\sigma_{\mathrm{loc}}}(0)$. Moreover, because $\dot{V}_{\mathrm{loc}}(z) < 0$ over $\mathcal{R}_{\mathrm{loc}} \setminus \{0\}$, any trajectory entering $\mathcal{R}_{\mathrm{loc}}$ remains therein for all $t$ and asymptotically converges to the origin. Finally, given that every trajectory originating from $\mathcal{R}_c$ enters $\mathcal{B}_\sigma(0)$ in finite time, the set inclusion $\mathcal{B}_\sigma(0) \subseteq \mathcal{R}_{\mathrm{loc}}$ provides a sufficient condition for asymptotic convergence to the origin. Noting that $\max_{\Vert{}z\Vert{}_2 \le \sigma} z^\top P z = \lambda_{\max}(P) \sigma^2$, the condition $\mathcal{B}_\sigma(0) \subseteq \mathcal{R}_{\mathrm{loc}}$ holds if $\lambda_{\max}(P) \sigma^2 < \lambda_{\min}(P) \sigma_{\mathrm{loc}}^2$, 
	which is equivalently formulated as
	\begin{equation} \label{eq:sigma_inclusion_bound}
		\sigma < \sigma_{\mathrm{threshold}} := \sigma_{\mathrm{loc}} \sqrt{\frac{\lambda_{\min}(P)}{\lambda_{\max}(P)}}.
	\end{equation} 
	Together with the SMT verification over $\mathcal{X}_{\mathrm{SMT}}$, condition \eqref{eq:sigma_inclusion_bound} completes the formal stability certification for the closed-loop system. Conversely, if \eqref{eq:sigma_inclusion_bound} is violated, the candidate controller is rejected, and the algorithm reverts to the training loop to re-synthesize $\psi(z)$ and $V_z(z)$. 
	The complete implementation of the proposed counterexample-guided synthesis and formal verification procedure is summarized in Algorithm 1.

	\begin{algorithm}[t]
		\caption{\normalsize Synthesis and SMT-Based Verification of Neural Controllers}
		\label{alg:synthesis}
		\begin{algorithmic}[1]
			\State \parbox[t]{0.92\linewidth}{\textbf{Input:} Identified model \eqref{eq:normalized_model_identified}, domain of interest $\mathcal{X}_z$, pre-training parameter $\beta$, positive margin $\varepsilon$, penalty hyperparameters $\gamma_1$ and $\gamma_2$, exclusion radius $\sigma$, precision tolerance $\delta$ and regularization parameter $\nu$.}
			\State \textbf{Output:} Verified neural controller $\psi(z)$, verified neural Lyapunov function $V_z(z)$.
			\State \textbf{Initialization:} Construct $\psi$ and $V_z$ with $\psi(0)=0$ and $V_z(0)=0$.
			\State \textbf{Pre-training:} Pre-train $V_z$ to approximate $V_{\text{target}}$.
			\State \parbox[t]{0.92\linewidth}{\textbf{Seed Sampling:} Sample initial dataset $\overline{\mathbb{D}} \subset \mathcal{X}_z \setminus \{0\}$.}
			\While{\textbf{true}}
			\For{$i = 1, 2, \dots$}
			\State \parbox[t]{0.78\linewidth}{Perform hybrid sampling targeting $\mathcal{X}_{\text{SMT}}$ and boundary $\partial\mathcal{B}_\sigma(0)$.}
			\State \parbox[t]{0.78\linewidth}{Search counterexamples $\mathbb{D}_{\text{ce}}$ via heuristic search paradigms.}
			\If{$\mathbb{D}_{\text{ce}} = \emptyset$}
			\State \textbf{break}
			\EndIf
			\State Augment training dataset: $\overline{\mathbb{D}} \leftarrow \overline{\mathbb{D}} \cup \mathbb{D}_{\text{ce}}$.
			\For{$\text{epoch} = 1, \dots, N_{\text{epoch}}$}
			\State \parbox[t]{0.70\linewidth}{Update $\psi$ and $V_z$ by minimizing $L_{\text{total}}(V_z, \psi)$ over $\overline{\mathbb{D}}$.}
			\EndFor
			\EndFor
			
			\State Invoke SMT solver to verify \eqref{stability.con.normalized.prop} over $\mathcal{X}_{\text{SMT}}$.
			\If{SMT verification fails}
			\State \parbox[t]{0.78\linewidth}{Extract counterexamples $\mathbb{D}_{\text{SMT}}$ from the returned falsifying interval.}
			\State Augment training dataset: $\overline{\mathbb{D}} \leftarrow \overline{\mathbb{D}} \cup \mathbb{D}_{\text{SMT}}$.
			\Else
			\State Evaluate the Jacobian matrix $A_{\text{cl}}$.
			\If{$A_{\text{cl}}$ is Hurwitz}
			\State \parbox[t]{0.70\linewidth}{Solve $P$, estimate $H_i$ and compute $\sigma_{\text{threshold}}$.}
			\If{$\sigma < \sigma_{\text{threshold}}$}
			\State \Return $\psi, V_z$.
			\EndIf
			\EndIf
			\EndIf
			\EndWhile
		\end{algorithmic}
	\end{algorithm}
	
	\begin{rem} \label{rem:sdp_objective_and_sampling}
		It is worth noting that the SDP \eqref{eq:lmi_sdp_opt} inherently optimizes $\sigma_{\text{threshold}}$. Specifically, the constraint $P \preceq I$ enforces $\|P\|_2 \le 1$ that suppresses $K_g$, while the optimization objective $\max \gamma + \nu \log \det(P)$ maximizes $\lambda_{\min}(Q)$ and drives $\sqrt{\frac{\lambda_{\min}(P)}{\lambda_{\max}(P)}}$ towards 1, thereby enlarging $\sigma_{\text{threshold}}$ to facilitate the satisfaction of \eqref{eq:sigma_inclusion_bound}. Furthermore, to compute $H_i$ ($i = 1, \dots, n$), a set of state samples $\{z_k\}_{k=1}^{\bar{N}}$ is densely drawn within $\mathcal{B}_{\sigma}(0)$. The Hessian spectral norm $\|\nabla^2 f_{\text{cl}, i}(z_k)\|_2$ at each sample point is evaluated using PyTorch, and $H_i$ is subsequently determined by taking the maximum value across all evaluated samples. \hfill $\blacksquare$
	\end{rem}
	
	\begin{rem}
		Unlike conventional deep learning paradigms reliant on static datasets, the proposed synthesis framework dynamically expands the training set $\overline{\mathbb{D}}$ through three sources: (i) initial seed samples drawn across $\mathcal{X}_z \setminus \{0\}$, (ii) targeted counterexamples generated via PGD, GAN, and PSO algorithms, and (iii) formal counterexamples extracted by the SMT solver during the verification phase. \hfill $\blacksquare$
	\end{rem}	
	
	\begin{rem}
		Verifying stability conditions near the boundary of $\mathcal{B}_{\sigma}(0)$ poses computational challenges for interval-based SMT solvers, rendering this transition zone prone to verification failures without targeted optimization. To mitigate this issue, we implement a hybrid sampling strategy specifically designed to initialize the PGD search trajectories, train the GAN mapping, and position the PSO swarm particles. Specifically, the total sample allocation is partitioned into a global component drawn from $\mathcal{X}_z \setminus \{0\}$ to ensure extensive baseline coverage, and a local component densely clustered along the boundary $\partial\mathcal{B}_\sigma(0)$. This targeted sampling strategy enhances the search efficacy of the PGD, GAN, and PSO algorithms in uncovering potential counterexamples along this critical interface, thereby facilitating the subsequent formal stability certification. \hfill $\blacksquare$
	\end{rem} 
	We demonstrate the proposed framework via a numerical example.
	\begin{example} \label{ex:lorenz} 
		Consider the dynamics of a controlled Lorenz chaotic system \cite{vincent1991control}:
		\begin{equation} \label{eq:lorenz_dynamics}
			\begin{aligned}
				\dot{x}_1 &= a (x_2 - x_1), \\
				\dot{x}_2 &= r x_1 - x_2 - x_1 x_3 + u, \\
				\dot{x}_3 &= x_1 x_2 - \eta x_3,
			\end{aligned}
		\end{equation}
		where $x := [x_1, x_2, x_3]^\top \in \mathbb{R}^3$ denotes the state vector, $u \in \mathbb{R}$ represents the control input, and $a, r, \eta > 0$ are physical system parameters. 
		The control objective is to stabilize the chaotic system at the origin. Note that the origin is an unstable equilibrium of the uncontrolled system for $r > 1$.
		
		Let $\mathcal{X} = \{x \in \mathbb{R}^3 \mid -b_i \le x_i \le b_i, \; i = 1, 2, 3\}$ denote the target stabilization domain, and $S = \operatorname{diag}(\frac{1}{b_1}, \frac{1}{b_2}, \frac{1}{b_3})$ be the state scaling matrix. By expressing \eqref{eq:lorenz_dynamics} in terms of the normalized state $z = Sx$, the transformed dynamics governing $z$ is formulated as:
		\begin{equation}
			\begin{aligned}
				\dot{z}_1 &= a \left( \frac{b_2}{b_1} z_2 - z_1 \right), \\
				\dot{z}_2 &= \frac{r b_1}{b_2} z_1 - z_2 - \frac{b_1 b_3}{b_2} z_1 z_3 + \frac{1}{b_2} u, \\
				\dot{z}_3 &= \frac{b_1 b_2}{b_3} z_1 z_2 - \eta z_3.
			\end{aligned}
		\end{equation} 
		Accordingly, we construct the normalized function library as $\Phi(z, u) = [z_1, z_2, z_3, u, z_1 z_3, z_1 z_2]^\top$.

		As the first step, we collect the dataset $\mathbb{D}$ for identification purpose from the system \eqref{eq:lorenz_dynamics} setting the parameters as $a = 10$, $r = 28$, and $\eta = \frac{8}{3}$. The domain bounds defining $\mathcal{X}$ are specified as $b_1 = b_2 = b_3= 1$. Here, we collect $T = 300$ samples by uniformly sampling the control input $u \in [-20, 20]$ and the state $x \in \mathcal{X}$. Based on the collected dataset $\mathbb{D}$, the identified normalized model \eqref{eq:normalized_model_identified} is constructed.

		To prepare for the joint synthesis and formal stability verification, the neural controller $\psi(z)$ and neural Lyapunov function $V_z(z)$ are configured as feedforward neural networks with one hidden layer of $80$ neurons and $\tanh$ activation function. To satisfy the equilibrium stability prerequisites, the bias-elimination architecture is implemented to strictly enforce $V_z(0) = 0$ and $\psi(0) = 0$. Prior to the training phase, a warm-up stage is executed wherein $V_z(z)$ is pre-trained to approximate $V_{\text{target}}(z) = 3 \|z\|_2^2$. Regarding the loss function, we set the convergence margin $\varepsilon = 0.01$, and the hyperparameters $\gamma_1 = 10$ and $\gamma_2 = 10$. The network parameters are optimized via the Adam optimizer with a learning rate of $1 \times 10^{-3}$.
		
		The training phase begins with the construction of the training dataset $\overline{\mathbb{D}}$. Initially, $\overline{\mathbb{D}}$ is populated with $50$ batches of samples, each with a batch size of $300$, generated via uniform state sampling across the domain $\mathcal{X}_z = [-1, 1]^3$. Let the exclusion ball be of radius $\sigma = 0.1$, and thus the SMT verification domain is $\mathcal{X}_{\text{SMT}} = \{z \in \mathbb{R}^3 \mid z \in [-1, 1]^3 \setminus \mathcal{B}_{\sigma}(0)\}$. Within each training loop, a counterexample search is conducted by deploying the PGD algorithm over $p = 60$ steps with a step size of $\kappa = 0.04$. A hybrid sampling strategy is applied to initialize the PGD candidate trajectories: $70\%$ of the initial states are drawn uniformly across $\mathcal{X}_z$, whereas the remaining $30\%$ are densely clustered along the boundary $\partial\mathcal{B}_{\sigma}(0)$. The uncovered counterexamples are aggregated and dynamically appended to $\overline{\mathbb{D}}$. For formal stability verification, the SMT solver dReal is deployed to validate the Lyapunov conditions over $\mathcal{X}_{\text{SMT}}$ with a precision tolerance of $\delta = 1 \times 10^{-4}$.

		In the numerical evaluation, dReal successfully returns a formal verification certificate, confirming that the Lyapunov conditions \eqref{stability.con.normalized.prop} hold universally across $\mathcal{X}_{\text{SMT}}$. Furthermore, regarding the local domain $\mathcal{B}_\sigma(0)$, the maximum real part of the eigenvalues of $A_{\text{cl}}$ is $-2.6667$, confirming its Hurwitz property. By setting $\nu = 0.1$, solving the SDP \eqref{eq:lmi_sdp_opt} and evaluating the Hessian bounds $H_i$ over $\mathcal{B}_{\sigma}(0)$ yield a threshold $\sigma_{\text{threshold}} = 0.8277$, which satisfies $\sigma < \sigma_{\mathrm{threshold}}$ and thereby ensures closed-loop convergence. By Proposition \ref{prop:equivalence}, the synthesized neural controller $u = \psi(Sx)$ formally guarantees the stability of the Lorenz system at the origin.

		To assess the control performance, we examine the 3D phase portraits of both the open-loop and closed-loop Lorenz systems, initialized from four distinct states in $\mathcal{X}$.
		As depicted in Figure~\ref{fig:lorenz_3d_comparison}, the open-loop state trajectories exhibit typical chaotic motion, persistently orbiting within the chaotic attractor without converging to the origin. In contrast, under the synthesized neural controller, all closed-loop state trajectories originating from these initial points are driven to the origin.
		
		\begin{figure}[ht!]
			\centering
			\begin{minipage}[b]{0.8\linewidth}
				\centering
				\includegraphics[width=\textwidth]{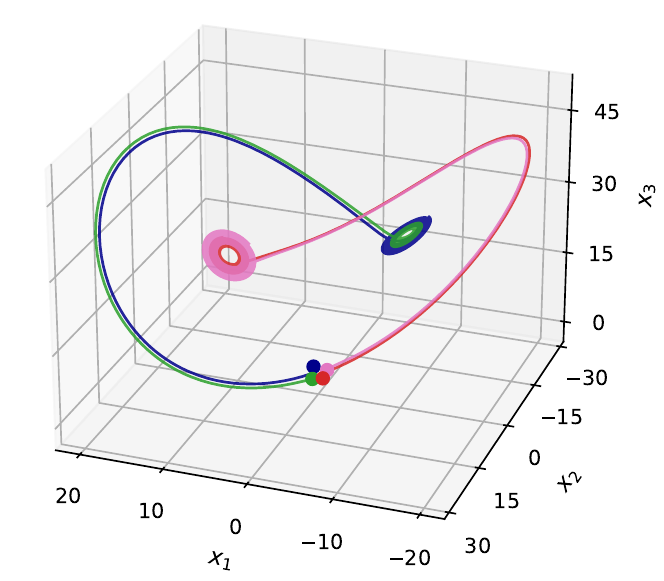}
				\\{\scriptsize (a) Open-loop trajectories.}
			\end{minipage}
			\begin{minipage}[b]{0.8\linewidth}
				\centering
				\includegraphics[width=\textwidth]{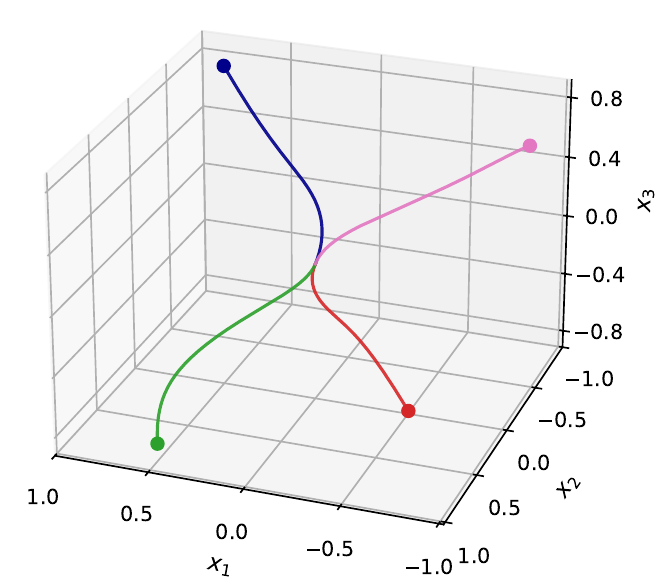}
				\\{\scriptsize (b) Closed-loop trajectories.}
			\end{minipage}
			\caption{3D phase portraits of the open-loop and closed-loop Lorenz systems, initialized from four distinct initial conditions in $\mathcal{X}$.}
			\label{fig:lorenz_3d_comparison}
		\end{figure}
		
	\end{example}

	\section{Input Constraints} \label{sec:input_constraints}
	In practical physical systems, actuators are invariably subject to saturation limits. To ensure safe controller deployment, the set of admissible control inputs $\mathcal{U} \subset \mathbb{R}^m$ is formulated as
	\begin{equation}
		\mathcal{U} := \{u \in \mathbb{R}^m \mid -u_{i,\min} \le u_i \le u_{i,\max}, \; i = 1, \dots, m\},
	\end{equation}
	where $u_{i,\min} > 0$ and $u_{i,\max} > 0$ denote the lower and upper bounds on $u_i$, respectively.

	To enforce input constraints, one intuitive approach is to integrate the input admissibility condition, i.e., $\phi(x) \in \mathcal{U}$, into the stability formulation \eqref{stability.con.free}, penalize its violations in the loss function during training, and formally verify it using SMT solvers. However, incorporating input constraints into SMT verification substantially increases computational burden, while penalizing input violations distracts the neural network's optimization focus on fulfilling the Lyapunov stability conditions. To circumvent these limitations, we incorporate a hard saturation architecture into the neural feedback design, which guarantees input constraints by construction. Specifically, we first construct a normalized controller $\bar{u} = \bar \psi(z)$ parameterized as:
	\begin{equation}
		\bar \psi(z) := \tanh \big( \mathcal{N}_u(z) - \mathcal{N}_u(0) \big).
	\end{equation}
	By passing the control signal through the $\tanh$ function, $\bar{u}$ is confined to the unit hypercube $[-1, 1]^m$, thereby serving as a normalized control input. The physical control input $u$ is subsequently recovered via the linear scaling transformation:
	\begin{equation}
		u = S_u \bar{u},
	\end{equation}
	where $S_u := \operatorname{diag}(b_{u,1}, b_{u,2}, \dots, b_{u,m})$ denotes the input scaling matrix with $b_{u,i} := \min(u_{i,\min}, u_{i,\max})$, which guarantees $u \in \mathcal{U}$ by construction.

	Analogous to \eqref{sys.d3}, the $z$-dynamics under the control input $\bar{u}$ are expressed as
	\begin{equation}
		\label{sys.d4}
		\dot z = \bar M \bar{\Phi}(z, \bar{u})
	\end{equation}
	where $\bar \Phi(z,\bar u) := T_{\Phi} \Phi(z, u)$ represents the input-normalized function library, and $T_{\Phi} \in \mathbb{R}^{s \times s}$ is a non-singular diagonal matrix. The normalized data matrix $\bar{\Phi}_0 \in \mathbb{R}^{s \times T}$ is organized from the dataset $\mathbb{D}$ as:
	\begin{equation}
		\begin{aligned}
			\bar \Phi_0 := \Big[ & \bar \Phi(S x(t_0),S_u^{-1} u(t_0)) \quad \bar \Phi(S x(t_1),S_u^{-1} u(t_1)) \\
			& \cdots \quad \bar \Phi(S x(t_{T-1}),S_u^{-1} u(t_{T-1})) \Big].
		\end{aligned}
	\end{equation}
	These matrices satisfy $Z_1 = \bar M \bar \Phi_0$, and $\bar \Phi_0$ preserves the full row rank property of $\Phi_0$. The matrix $\bar M$ is identified as $\bar M = Z_1 \bar \Phi_0^\dagger$, where $\bar \Phi_0^\dagger \in \mathbb{R}^{T \times s}$ denotes the right inverse of $\bar \Phi_0$. Hence, the identified model under input normalization is expressed as
	\begin{equation} \label{eq:normalized_model_identified2}
		\dot{z} = Z_1 \bar \Phi_0^\dagger \bar \Phi(z,\bar u).
	\end{equation}
	As formalized below, the stability of \eqref{eq:normalized_model_identified2} is equivalent to that of the system \eqref{sys.d3}.
	\begin{proposition} \label{prop:input_norm_equivalence}
		Consider the system \eqref{sys.d3} and the input-normalized system \eqref{eq:normalized_model_identified2} under the input transformation $u = S_u \bar{u}$. Suppose there exist $\bar{\psi}$ and $V_z$ satisfying $\bar{\psi}(0)=0$, $V_z(0)=0$, and $V_z(z)>0$, $\forall z \in \mathcal{X}_z \setminus \{0\}$. Then, for each $z \in \mathcal{X}_z \setminus \{0\}$, the condition
		\begin{equation} \label{eq:stability_cond_normalized}
			\frac{\partial V_z(z)}{\partial z} Z_1 \bar{\Phi}_0^{\dagger} \bar{\Phi}(z, \bar{\psi}(z)) < 0
		\end{equation}
		holds if and only if 
		\begin{equation} \label{eq:stability_cond_original}
			\frac{\partial V_z(z)}{\partial z} Z_1 \Phi_0^{\dagger} \Phi(z, S_u \bar{\psi}(z)) < 0
		\end{equation}
		holds.
	\end{proposition}
	\begin{pf}
		Since $\bar \Phi(z,\bar u) = T_{\Phi} \Phi(z, u)$ and $T_{\Phi}$ is a non-singular diagonal matrix, the constructed data matrices satisfy $\bar{\Phi}_0 = T_{\Phi} \Phi_0$, which yields
		\begin{equation*}
			\bar{\Phi}_0^{\dagger} = (T_{\Phi} \Phi_0)^{\dagger} = \Phi_0^{\dagger} T_{\Phi}^{-1}.
		\end{equation*}
		Furthermore, under the controller $\bar{u} = \bar{\psi}(z)$ with the input transformation $u = S_u \bar{u}$, we have $\bar{\Phi}(z, \bar{\psi}(z)) = T_{\Phi} \Phi(z, S_u \bar{\psi}(z))$. Substituting these relations yields
		\begin{equation*}
			\begin{aligned}
				Z_1 \bar{\Phi}_0^{\dagger} \bar{\Phi}(z, \bar{\psi}(z)) 
				&= Z_1 \left( \Phi_0^{\dagger} T_{\Phi}^{-1} \right) \left( T_{\Phi} \Phi(z, S_u \bar{\psi}(z)) \right) \\
				&= Z_1 \Phi_0^{\dagger} \left( T_{\Phi}^{-1} T_{\Phi} \right) \Phi(z, S_u \bar{\psi}(z)) \\
				&= Z_1 \Phi_0^{\dagger} \Phi(z, S_u \bar{\psi}(z)).
			\end{aligned}
		\end{equation*}
		Left-multiplying both sides by $\frac{\partial V_z(z)}{\partial z}$, we obtain
		\begin{equation*}
			\frac{\partial V_z(z)}{\partial z} Z_1 \bar{\Phi}_0^{\dagger} \bar{\Phi}(z, \bar{\psi}(z)) = \frac{\partial V_z(z)}{\partial z} Z_1 \Phi_0^{\dagger} \Phi(z, S_u \bar{\psi}(z)),
		\end{equation*}
		which holds pointwise for all $z \in \mathcal{X}_z \setminus \{0\}$. Consequently, condition \eqref{eq:stability_cond_normalized} holds over $\mathcal{X}_z \setminus \{0\}$ if and only if condition \eqref{eq:stability_cond_original} holds over $\mathcal{X}_z \setminus \{0\}$, which completes the proof. \hfill$\square$
	\end{pf}

	By virtue of the joint equivalences established in Propositions \ref{prop:equivalence} and \ref{prop:input_norm_equivalence}, the synthesis objective under input constraints reduces to designing a neural controller $\bar{u} = \bar{\psi}(z)$ and a neural Lyapunov function $V_z(z)$ that satisfy \eqref{stability.prop.norm1} and \eqref{eq:stability_cond_normalized}. Accordingly, the loss term $L_{\text{ND}}$ in \eqref{loss.nd} is adapted to account for the input-normalized dynamics:
	\begin{equation}\label{loss.ND_sat}
		\begin{aligned}
			L_{\text{ND}}(V_z, \bar{\psi}) = \frac{1}{N} \sum_{k=1}^{N} \operatorname{ReLU} & \left( \frac{\partial V_z(z_k)}{\partial z} Z_1 \bar{\Phi}_0^{\dagger} \bar{\Phi}(z_k, \bar{\psi}(z_k)) \right. \\
			& \quad \left. + \varepsilon \cdot \Omega(z_k) \right).
		\end{aligned}
	\end{equation}
	Consequently, the stability conditions \eqref{stability.con.free} are guaranteed to hold across $\mathcal{X} \setminus \{0\}$ by recovering the physical controller as $\phi(x) = S_u \bar{\psi}(Sx)$ and establishing the physical Lyapunov function as $V(x) = V_z(Sx)$.

	\begin{rem}
		Appending a $\tanh$ function to the raw MLP output provides an intuitive yet effective strategy for both enforcing input constraints and achieving input normalization. From an optimization perspective, in the absence of input normalization, the gradient updates of the neural controller tend to dominate those of the neural Lyapunov function when $b_{u,i} \gg 1$, whereas the converse holds when $b_{u,i} \ll 1$. Hence, input normalization plays a critical role in balancing the co-optimization of both networks. \hfill $\blacksquare$
	\end{rem}

	\begin{example} \label{ex:maglev} 
		Consider the dynamics of a magnetic levitation system \cite{ortega2001putting}:
		\begin{equation} \label{eq:maglev_dynamics}
			\begin{aligned}
				\dot{\lambda} &= -\frac{R}{k} \lambda h + u, \\
				\dot{h} &= \frac{1}{m} w, \\
				\dot{w} &= -\frac{1}{2k} \lambda^2 + mg, 
			\end{aligned}
		\end{equation}
		where $\lambda$ is the coil flux, $h$ is the absolute air gap, $w$ is the momentum of the levitated iron ball, and $u$ is the control voltage applied to the electromagnet coil. Let $x := [\lambda, h, w]^\top$ be the state vector. The control objective is to stabilize the iron ball at a desired height $h_*$. The corresponding steady-state equilibrium point is given by $x_* = [\lambda_*, h_*, 0]^\top$ and $u_* = \frac{R h_* \lambda_*}{k}$, where $\lambda_* = \sqrt{2kmg}$. To shift the equilibrium to the origin, we introduce the error state and input variables as $\tilde{x} := x - x_*$ and $\tilde{u} := u - u_*$, respectively. Note that $\tilde{x} = 0$ is an unstable equilibrium of the uncontrolled system.

		Let $\mathcal{X} = \{\tilde{x} \in \mathbb{R}^3 \mid -b_i \le \tilde{x}_i \le b_i, \; i = 1, 2, 3\}$ denote the target stabilization domain, and $S = \operatorname{diag}(\frac{1}{b_1}, \frac{1}{b_2}, \frac{1}{b_3})$ be the state scaling matrix. Furthermore, let $\mathcal{U} = \{\tilde{u} \in \mathbb{R} \mid -u_{\min} \le \tilde{u} \le u_{\max}\}$ be the set of admissible control inputs, and $S_u = b_u$ be the input scaling factor with $b_u = \min(u_{\min}, u_{\max})$. By expressing \eqref{eq:maglev_dynamics} in terms of the shifted, normalized state $z = S\tilde{x}$ and control input $\bar{u} = S_u^{-1} \tilde{u}$, the transformed dynamics governing $z$ are formulated as:
		\begin{equation}
			\begin{aligned}
				\dot{z}_1 &= -\frac{R}{k} \Big(h_* z_1 + \frac{\lambda_* b_2}{b_1} z_2 + b_2 z_1 z_2 \Big) + \frac{b_u}{b_1} \bar{u}, \\
				\dot{z}_2 &= \frac{b_3}{b_2 m} z_3, \\
				\dot{z}_3 &= -\frac{1}{2 k b_3} \left( 2 \lambda_* b_1 z_1 + b_1^2 z_1^2 \right).
			\end{aligned}
		\end{equation}
		Accordingly, we construct the normalized function library as $\Phi(z, \bar{u}) = [z_1, z_2, z_3, \bar{u}, z_1 z_2, z_1^2]^\top$.
		
		As the first step, we collect the dataset $\mathbb{D}$ for identification purpose from the system \eqref{eq:maglev_dynamics} setting the parameters as $R = 2.52$, $m = 0.0844$, $g = 9.81$, $k = 6.4042 \times 10^{-5}$, and the target equilibrium height as $h_* = 0.007$. The domain bounds defining $\mathcal{X}$ are set to $b_1 = 0.0015$, $b_2 = 0.0015$, and $b_3 = 0.008$, and the input bounds defining $\mathcal{U}$ are set to $u_{\max} = u_{\min} = 2$. Given the small magnitude of the physical bounds $b_i \ll 1$ ($i = 1, 2, 3$), state normalization is essential; otherwise, training directly on the unscaled error state $\tilde{x}$ might induce numerical issues. Here, we collect $T = 50$ samples by uniformly sampling the control input $u$ and state $x$ such that $\tilde{u} \in \mathcal{U}$ and $\tilde{x} \in \mathcal{X}$, respectively. Based on the collected dataset $\mathbb{D}$, the identified model \eqref{eq:normalized_model_identified2} is constructed.
		
		For network configuration, both $V_z(z)$ and $\bar{\psi}(z)$ are parameterized using feedforward neural networks with one hidden layer of $80$ neurons and $\tanh$ activation function. To satisfy the equilibrium stability prerequisites, the bias-elimination architecture \eqref{equilibrium.constraint} is implemented to guarantee $V_z(0) = 0$ and $\bar{\psi}(0) = 0$, with an additional $\tanh$ function enforcing the input constraint $\tilde{u} \in \mathcal{U}$. Prior to the training phase, a warm-up stage is executed wherein $V_z(z)$ is pre-trained to approximate $V_{\text{target}}(z) = 6 \|z\|_2^2$. Regarding the loss function, we set the convergence margin $\varepsilon = 0.01$, and the hyperparameters $\gamma_1 = 100$ and $\gamma_2 = 25$ to balance the penalty terms. The network parameters are updated via the Adam optimizer, and $V_z(z)$ is trained with a learning rate of $1 \times 10^{-3}$, whereas $\bar{\psi}(z)$ is trained with a learning rate of $5 \times 10^{-4}$ together with a weight decay of $1 \times 10^{-3}$.

		The training phase begins with the construction of the training dataset $\overline{\mathbb{D}}$. Initially, $\overline{\mathbb{D}}$ is populated with $50$ batches of samples, each with a batch size of $300$, generated via uniform state sampling across the domain $\mathcal{X}_z$. Let the exclusion ball be of radius $\sigma = 0.04$, and thus the SMT verification domain is $\mathcal{X}_{\text{SMT}} = \{z \in \mathbb{R}^3 \mid z \in [-1, 1]^3 \setminus \mathcal{B}_{\sigma}(0)\}$. Within each training loop, a synergistic counterexample search is executed by deploying both PGD and PSO. Specifically, the gradient-based PGD search initializes a pool of $100$ batches of candidate samples, which are recursively updated over $p = 60$ steps with a step size of $\kappa = 0.04$. Furthermore, to bypass saturation-induced vanishing gradients, the derivative-free PSO search deploys a swarm of $N_{\mathrm{pso}} = 1000$ particles over $p = 60$ steps, configured with $\omega = 0.5$ and $c_1 = c_2 = 1.5$. Both the PGD trajectories and the PSO swarms adopt the hybrid sampling configuration detailed in Example \ref{ex:lorenz}. The counterexamples uncovered by both search algorithms are aggregated and dynamically appended to $\overline{\mathbb{D}}$. In the formal stability verification phase, the SMT solver dReal is deployed to validate the Lyapunov conditions \eqref{stability.prop.norm1} and \eqref{eq:stability_cond_normalized} over $\mathcal{X}_{\text{SMT}}$ with a precision tolerance of $\delta = 1 \times 10^{-4}$.

		In the numerical evaluation, dReal successfully returns a formal verification certificate, confirming that the Lyapunov conditions \eqref{stability.prop.norm1} and \eqref{eq:stability_cond_normalized} hold universally across $\mathcal{X}_{\text{SMT}}$. Furthermore, regarding the local domain $\mathcal{B}_\sigma(0)$, the maximum real part of the eigenvalues of $A_{\text{cl}}$ is $-7.4789$, confirming its Hurwitz property. By setting $\nu = 0.1$, solving the SDP \eqref{eq:lmi_sdp_opt} and evaluating the Hessian bounds $H_i$ over $\mathcal{B}_{\sigma}(0)$ yield a threshold $\sigma_{\text{threshold}} = 0.0617$, which satisfies $\sigma < \sigma_{\mathrm{threshold}}$ and thereby ensures closed-loop convergence. By Propositions \ref{prop:equivalence} and \ref{prop:input_norm_equivalence}, the synthesized neural controller $u = u_* + S_u \bar{\psi}\big(S(x - x_*)\big)$ formally guarantees to stabilize the iron ball at the desired height $h_*$. The closed-loop state and input evolutions of the physical dynamics \eqref{eq:maglev_dynamics} initialized at $x(0) = [\lambda_*, 0.008, 0]^\top$ are displayed in Figure~\ref{fig:maglev_simulation}, where the state trajectories converge to the target equilibrium point $x_*$ and the control input remains within the set $\mathcal{U}$.
		
		For the closed-loop dynamics, we numerically determine the set $\mathcal{H} := \{z \in \mathbb{R}^3 \mid \dot V_z(z) < 0\}$, with $\dot V_z(z) = \frac{\partial V_z(z)}{\partial z} Z_1 \bar{\Phi}_0^{\dagger} \bar{\Phi}(z, \bar{\psi}(z))$, over which the Lyapunov function $V_z(z)$ decreases, and any sub-level set $\mathcal{R}_c = \{z \in \mathbb{R}^3 \mid V_z(z) \le c\}$ with $c > 0$ contained within $\mathcal{H} \cup \{0\}$ yields an estimate of the ROA for the normalized closed-loop system. Figure~\ref{fig:maglev_simulation2} shows the planar projections of the set $\mathcal{H}$, a sublevel set of $V_z(z)$, and the target domain $\mathcal{X}_z$. Note that the green set ($\mathcal{H}$) is larger than the projection of any sub-level set of $V_z(z)$ contained in $\mathcal{X}_z$, which indicates that the obtained neural controller achieves a larger region of attraction than theoretically guaranteed.
		
		\begin{figure}[ht!]
			\centering
			\includegraphics[width=\columnwidth]{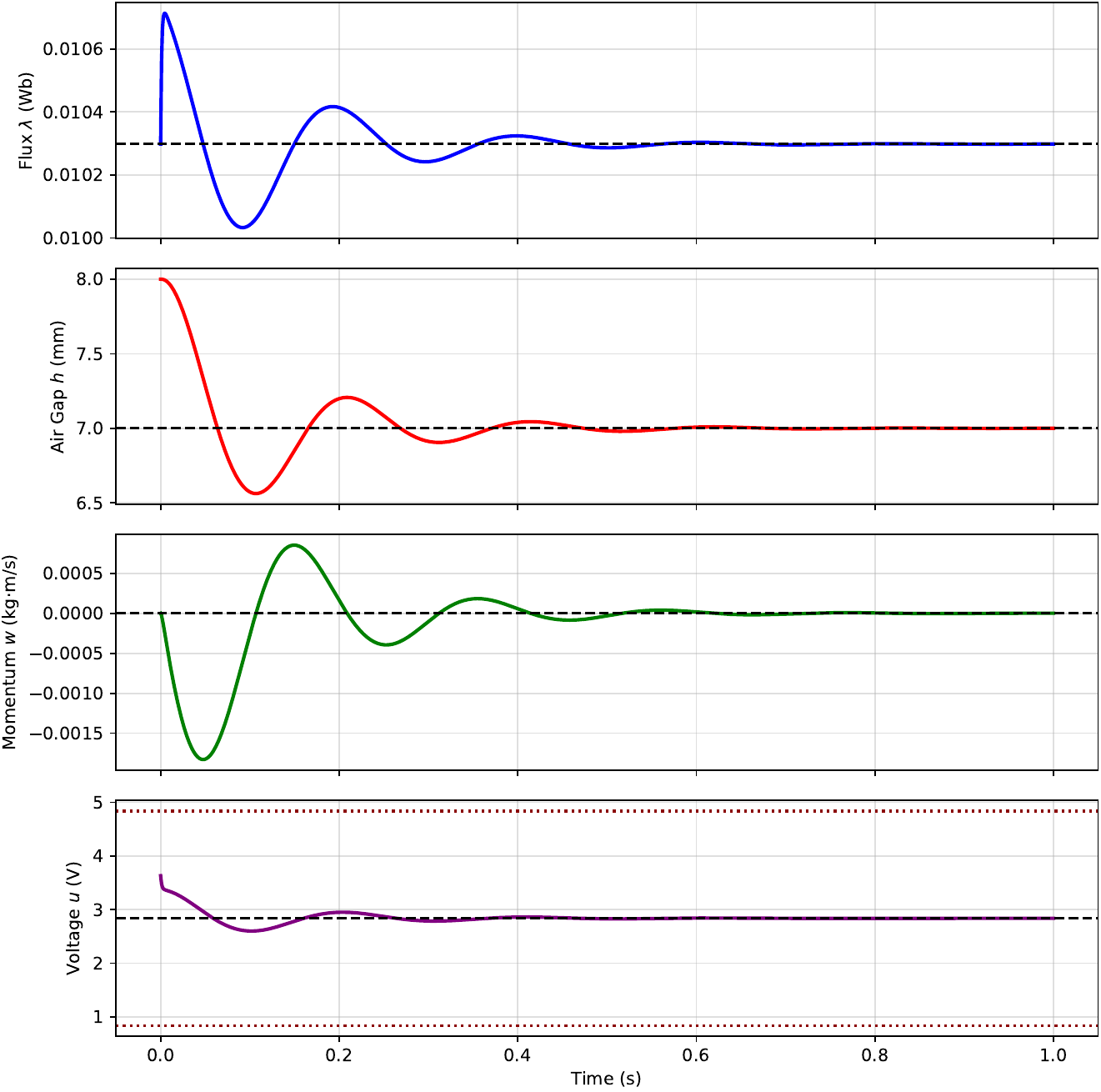}
			\caption{State and control responses of the closed-loop magnetic levitation system, initialized at state $x(0) = [\lambda_*, 0.008, 0]^\top$.}
			\label{fig:maglev_simulation}
		\end{figure}
		
		\begin{figure}[ht!]
			\centering
			\includegraphics[width=\columnwidth]{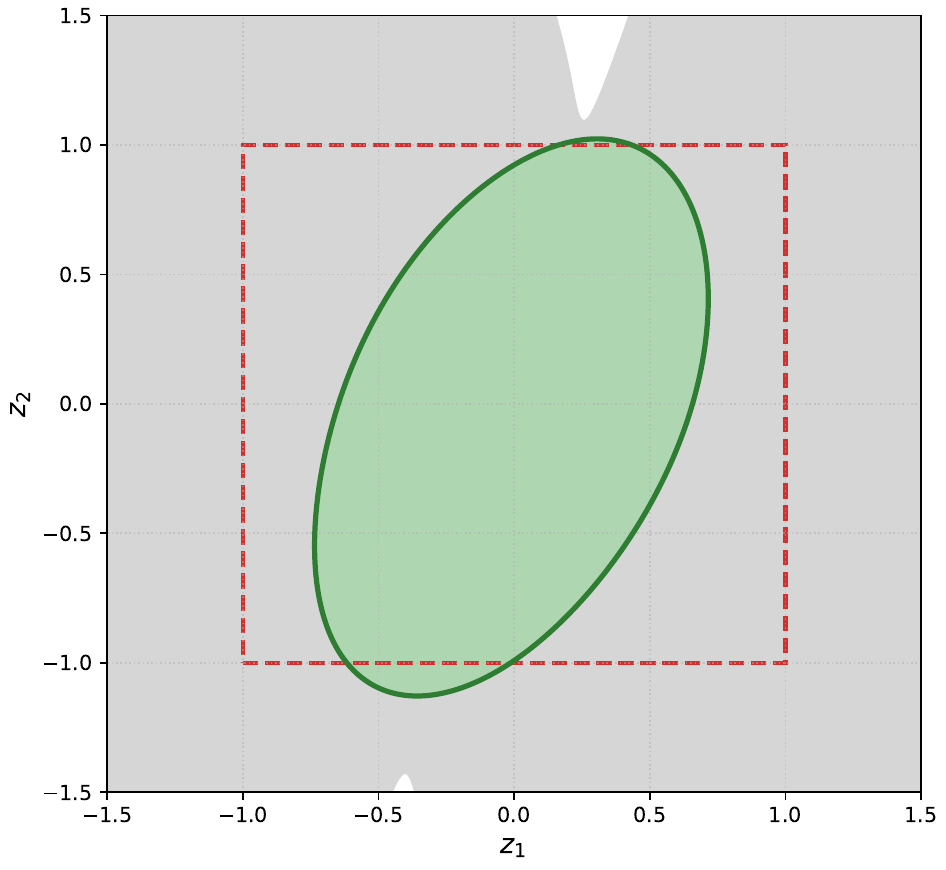}
			\caption{The gray set, the blue set, and the red dashed box represent the projections of the set $\mathcal{H}$, the Lyapunov sublevel set $\mathcal{R}_c$ with $c = 2.0621$, and the domain $\mathcal{X}_z$ onto the plane $\{z \in \mathbb{R}^3 \mid z_3 = 0\}$, respectively.}
			\label{fig:maglev_simulation2}
		\end{figure}
		
	\end{example}
	
	\section{Robustness to Disturbances} \label{sec:robustness}
	In this section, we extend the previous results to scenarios where additive perturbations corrupt the identification data. Note that offline data acquisition during the system identification phase is the sole step that interfaces with the physical plant, rendering the identification data susceptible to environmental disturbances. In turn, the subsequent controller synthesis and verification phase is conducted within a deterministic numerical framework, which is insulated from physical measurement noise. Hence, the exclusive source of uncertainty is the \emph{parametric} mismatch of the identified system matrix. Building upon this, we develop robust control design methods that explicitly account for the uncertainties arising from the corrupted identification data by enforcing the Lyapunov stability conditions over the entire set of data-consistent models.
	
	\subsection{Robust Stability Conditions under Data Perturbations}
	Consider the scenario where additive perturbations corrupt the state derivative measurement and the collected samples satisfy:
	\begin{equation}
		\dot{x}(t_{i}) = A Z(x(t_{i}), u(t_{i})) + d(t_{i}), \quad i = 0, \dots, T-1
	\end{equation}
	where $d(t_{i}) \in \mathbb{R}^{n}$ represents the unknown perturbation at sampling time $t_i$. Accordingly, the data matrices $X_1$ and $Z_0$ defined in \eqref{eq:data} satisfy the following relation:
	\begin{equation}
		X_{1} = A Z_{0} + D_{0},
	\end{equation}
	where $D_{0} := [d(t_{0}), d(t_{1}), \dots, d(t_{T-1})] \in \mathbb{R}^{n \times T}$ is the unknown data perturbation matrix. To perform robust synthesis, we formalize our prior knowledge regarding $D_{0}$ through the following bounded uncertainty set assumption.
	\begin{assumption} \label{ass:noise}
		Given a matrix $\Delta \in \mathbb{R}^{n \times n}$, $D_0 \in \mathcal{D} := \{ D \in \mathbb{R}^{n \times T} \colon DD^\top \preceq \Delta \Delta^\top \}$.
	\end{assumption}
	This formulation provides a versatile framework capable of accommodating various perturbation classes, including pointwise norm bounds, energy constraints, and stochastic disturbances. Under Assumption~\ref{ass:Z0}, the true system matrix can be expressed as $A = (X_{1} - D_{0}) Z_{0}^{\dagger}$. Since $D_0$ is unknown, we aim to ensure that the Lyapunov stability conditions hold robustly across the entire set of system matrices compatible with the corrupted data, leading to the following robust stability guarantee.
	\begin{theorem} \label{thm:robust_data}
		Consider the nonlinear system \eqref{sys.d}. Let Assumptions \ref{ass:f}, \ref{ass:Z0}, and \ref{ass:noise} hold. Let $\mathcal{X} \subseteq \mathbb{R}^n$ be a set containing the origin. Suppose there exist a scalar $\alpha > 0$, and neural networks $\phi$ and $V$ satisfying $\phi(0) = 0$ and $V(0) = 0$ such that for each $x \in \mathcal{X} \setminus \{0\}$:
		\begin{subequations}\label{eq:robust_stability}
			\begin{align}
				& V(x) > 0, \label{eq:robust_pd}\\
				& 2 \frac{\partial V(x)}{\partial x} X_1 Z_0^\dag Z(x,\phi(x)) \nonumber \\
				& \quad + \alpha \frac{\partial V(x)}{\partial x} \Delta \Delta^\top \left(\frac{\partial V(x)}{\partial x}\right)^\top \nonumber \\
				& \quad + \frac{1}{\alpha} Z(x,\phi(x))^\top {Z_0^\dag}^\top Z_0^\dag Z(x,\phi(x)) < 0. \label{eq:robust_nd}
			\end{align}
		\end{subequations}
		Then, the origin is asymptotically stable for the closed-loop system \eqref{eq:closed}.
	\end{theorem}
	\begin{pf}
		The proof relies on demonstrating that $V(x)$ serves as a valid Lyapunov function for all consistent systems of the form $\dot{x} = (X_{1} - D) Z_{0}^{\dagger} Z(x, \phi(x))$, $\forall D \in \mathcal{D}$. This requires verifying two properties: $(i)$ $V$ is positive definite, and $(ii)$ its time derivative along the system trajectories satisfies
		\begin{equation}\label{eq:lyapunov}
			\frac{\partial V(x)}{\partial x} (X_1 - D) Z_0^\dag Z(x,\phi(x)) < 0, \quad \forall x \in \mathcal{X} \setminus \{0\}, \; \forall D \in \mathcal{D}.
		\end{equation}
		The positive definiteness of $V$ is a direct consequence of $V(0) = 0$ together with \eqref{eq:robust_pd}. We then demonstrate that \eqref{eq:lyapunov} holds by virtue of \eqref{eq:robust_nd}. By Assumption \ref{ass:noise}, for each $x \in \mathcal{X} \setminus \{0\}$, \eqref{eq:robust_nd} gives
		\begin{align*}
			& 2 \frac{\partial V(x)}{\partial x} X_1 Z_0^\dag Z(x,\phi(x)) + \alpha \frac{\partial V(x)}{\partial x} D D^\top \left(\frac{\partial V(x)}{\partial x}\right)^\top \\
			& \quad + \frac{1}{\alpha} Z(x,\phi(x))^\top {Z_0^\dag}^\top Z_0^\dag Z(x,\phi(x)) < 0, \quad \forall D \in \mathcal{D}.
		\end{align*}
		Applying Young's inequality yields
		\begin{align*}
			& 2 \frac{\partial V(x)}{\partial x} X_1 Z_0^\dag Z(x,\phi(x)) - \frac{\partial V(x)}{\partial x} D Z_0^\dag Z(x,\phi(x)) \\
			& \quad - Z(x,\phi(x))^\top {Z_0^\dag}^\top D^\top \left(\frac{\partial V(x)}{\partial x}\right)^\top < 0, \quad \forall D \in \mathcal{D},
		\end{align*}
		which implies that \eqref{eq:lyapunov} holds. This completes the proof. \hfill$\square$
	\end{pf}
	\begin{rem} \label{rem:robust_stability_analysis}
		Note that the feasibility of \eqref{eq:robust_stability} hinges on the existence of $V(x)$ and $\phi(x)$ that render the nominal closed-loop term $2 \frac{\partial V(x)}{\partial x} X_1 Z_0^\dag Z(x,\phi(x))$ sufficiently negative to dominate both the data perturbation penalty $\alpha \frac{\partial V(x)}{\partial x} \Delta \Delta^\top \big(\frac{\partial V(x)}{\partial x}\big)^\top$ and the regressor-dependent term $\frac{1}{\alpha} Z(x,\phi(x))^\top {Z_0^\dag}^\top Z_0^\dag Z(x,\phi(x))$. Hence, the controllability of the nominal system plays a fundamental role in achieving this goal. \hfill $\blacksquare$
	\end{rem}
	\subsection{Robust Loss Design and Joint Synthesis}
	To translate the robust stability conditions \eqref{eq:robust_stability} into the penalty terms within loss function \eqref{loss.total}, the term $L_{\text{ND}}$ in \eqref{loss.nd} is modified to explicitly enforce \eqref{eq:robust_nd}, yielding the robust penalty term
	 $L_{\text{ND}}^{\text{R,d}}$:
	\begin{equation}
		\begin{aligned}
			& L_{\text{ND}}^{\text{R,d}}(V_z, \psi, \alpha) \\
			& = \frac{1}{N} \sum_{k=1}^{N} \operatorname{ReLU} \Bigg( 2 \frac{\partial V_z(z_k)}{\partial z} Z_1 \Phi_0^{\dagger} \Phi(z_k, \psi(z_k)) \\
			& \quad + \alpha \frac{\partial V_z(z_k)}{\partial z} \Delta_z \Delta_z^\top \left(\frac{\partial V_z(z_k)}{\partial z}\right)^\top \\
			& \quad + \frac{1}{\alpha} \Phi(z_k, \psi(z_k))^\top (\Phi_0^{\dagger})^\top \Phi_0^{\dagger} \Phi(z_k, \psi(z_k)) + \varepsilon \Omega(z_k) \Bigg).
		\end{aligned}
	\end{equation}
	Under the state normalization $z = Sx$, the normalized perturbation matrix is defined as $D_{z,0} := S D_0$, which satisfies $D_{z,0} D_{z,0}^\top \preceq \Delta_z \Delta_z^\top$ with $\Delta_z := S \Delta$ denoting the scaled perturbation bound matrix.

	Rather than treating $\alpha$ as an offline hyperparameter selected prior to the controller synthesis stage, our framework considers it as a \emph{decision variable} optimized jointly with $\psi(z)$ and $V_z(z)$. To achieve this, we re-parameterize $\alpha$ as:
	\begin{equation} \label{eq:alpha_reparam}
		\alpha = e^{\theta_{\alpha}},
	\end{equation}
	where $\theta_{\alpha} \in \mathbb{R}$ is an unconstrained trainable parameter updated via backpropagation. This exponential mapping strictly guarantees $\alpha > 0$ by construction, thereby fulfilling the positivity prerequisite of Theorem \ref{thm:robust_data}. During joint training, the optimizer automatically adjusts $\theta_{\alpha}$ to dynamically balance the terms within the loss function. Upon convergence, the learned parameter evaluates to a deterministic constant $\alpha^* = e^{\theta_{\alpha}^*}$.

	\begin{rem}[Joint Optimization of $\varepsilon$]
		\label{rem:joint_convergence_rate_opt}
		Analogous to the treatment of $\alpha$, rather than pre-selecting the margin parameter $\varepsilon$ \textit{a priori}, it can be re-parameterized as:
		\begin{equation}
			\varepsilon = e^{\theta_\varepsilon},
			\label{eq:epsilon_reparam}
		\end{equation}
		where $\theta_\varepsilon \in \mathbb{R}$ is a trainable parameter updated jointly during training. Note that unconstrained minimization of the loss function tends to yield a trivial solution $\varepsilon \to 0^+$, which dampens the closed-loop convergence rate. To circumvent this, a regularization term $\mathcal{L}_{\text{CR}}(\theta_\varepsilon) := -\lambda_\varepsilon \theta_\varepsilon$ with $\lambda_\varepsilon > 0$ is incorporated into the overall loss $L_{\text{total}}(V_z, \psi, \alpha, \varepsilon)$. The term $\mathcal{L}_{\text{CR}}(\theta_\varepsilon)$ serves to maximize $\varepsilon$ while maintaining the Lyapunov stability constraints $V_z(z) \ge \varepsilon \|z\|^2_2$ and $\dot{V}_z(z) \le -\varepsilon \|z\|^2_2$. \hfill $\blacksquare$
	\end{rem}
	\subsection{Robust Local Stability Certification}
	In the presence of state derivative measurement noise, the uncertain normalized closed-loop dynamics are given by
	\begin{equation} \label{eq:closed_loop_dynamics}
		\dot{z} = f_{\mathrm{cl}, D}(z) := (Z_1 - D_{z,0}) \Phi_0^{\dagger} \Phi(z, \psi(z)).
	\end{equation}
	Let $J_{\Phi}(0) := \left. \frac{\partial \Phi(z,\psi(z))}{\partial z} \right|_{z=0} \in \mathbb{R}^{s \times n}$ denote the Jacobian of $\Phi(z,\psi(z))$ evaluated at the origin. The first-order Taylor's expansion of \eqref{eq:closed_loop_dynamics} around $z=0$ is expressed as
	\begin{equation} \label{eq:robust_closed_loop_taylor}
		\dot{z} = A_{\mathrm{cl}, D} z + g(z, D_{z,0}),
	\end{equation}
	where $A_{\mathrm{cl}, D} := A_0 - D_{z,0} \Phi_{0}^{\dagger} J_{\Phi}(0)$ with $A_0 := Z_{1} \Phi_{0}^{\dagger} J_{\Phi}(0)$ representing the nominal closed-loop Jacobian, and $g(z, D_{z,0})$ denotes the higher-order remainder satisfying $g(0, D_{z,0}) = 0$ and $\left. \frac{\partial g(z, D_{z,0})}{\partial z} \right|_{z=0} = 0$.
	
	To establish local asymptotic stability robustly against all perturbation matrices $D_{z,0} \in \mathcal{D}_z$, we consider a candidate local Lyapunov function $V_{\text{loc}}(z) = z^\top P z$ with $P \in \mathbb{S}^{n \times n}$. The time derivative of $V_{\text{loc}}(z)$ along the trajectories of \eqref{eq:robust_closed_loop_taylor} is given by
	\begin{equation} \label{eq:robust_v_dot_expansion}
		\dot{V}_{\text{loc}}(z) = z^\top \left( A_{\mathrm{cl}, D}^\top P + P A_{\mathrm{cl}, D} \right) z + 2 z^\top P g(z, D_{z,0}).
	\end{equation}
	By Young's inequality, the quadratic term in \eqref{eq:robust_v_dot_expansion} can be upper-bounded for any scalar $\varphi > 0$ as
	\begin{equation} \label{eq:robust_quadratic_bound}
		\begin{aligned}
			& z^\top \big( A_{\mathrm{cl}, D}^\top P + P A_{\mathrm{cl}, D} \big) z \\
			& = z^\top \Big( A_0^\top P + P A_0 - J_{\Phi}(0)^\top (\Phi_{0}^{\dagger})^\top D_{z,0}^\top P \\
			& \quad - P D_{z,0} \Phi_{0}^{\dagger} J_{\Phi}(0) \Big) z \\
			& \le z^\top \Big( A_0^\top P + P A_0 + \frac{1}{\varphi} P \Delta_z \Delta_z^\top P \\
			& \quad + \varphi J_{\Phi}(0)^\top (\Phi_{0}^{\dagger})^\top \Phi_{0}^{\dagger} J_{\Phi}(0) \Big) z.
		\end{aligned}
	\end{equation}
	We then search for $P \succ 0$, $\varphi > 0$, and $\gamma > 0$ by solving the following SDP:
	\begin{equation} \label{eq:robust_lmi_sdp_opt}
		\begin{aligned}
			\max_{P \succ 0, \varphi > 0, \gamma > 0} \quad & \gamma + \nu \log \det(P) \\
			\text{s.t.} \quad & P \preceq I, \\
			& \begin{bmatrix} M_{11} & P \Delta_z \\ \ast & -\varphi I_n \end{bmatrix} \preceq 0,
		\end{aligned}
	\end{equation}
	where $\nu > 0$ is a regularization parameter and $M_{11} := A_0^\top P + P A_0 + \varphi J_{\Phi}(0)^\top (\Phi_{0}^{\dagger})^\top \Phi_{0}^{\dagger} J_{\Phi}(0) + \gamma I$.
	By Schur complement, the feasibility of \eqref{eq:robust_lmi_sdp_opt} guarantees that $z^\top \left( A_{\mathrm{cl}, D}^\top P + P A_{\mathrm{cl}, D} \right) z \le -\gamma \|z\|_2^2$ holds robustly for all $D_{z,0} \in \mathcal{D}_z$. If the SDP \eqref{eq:robust_lmi_sdp_opt} is infeasible, the candidate controller is rejected, and the algorithm reverts to the training loop to re-synthesize $\psi(z)$ and $V_z(z)$.
	
	On the other hand, if the SDP \eqref{eq:robust_lmi_sdp_opt} is feasible, by the Cauchy-Schwarz inequality together with Taylor's theorem, the higher-order cross-term in \eqref{eq:robust_v_dot_expansion} satisfies
	\begin{equation} \label{eq:robust_higher_order_cross_term_bound}
		2 z^\top P g(z, D_{z,0}) \le \sum_{i=1}^n \|P_{i,\cdot}\|_2 H_{i, D} \|z\|_2^3,
	\end{equation}
	where $H_{i, D} := \max_{D_{z,0} \in \mathcal{D}_z} \max_{z \in \mathcal{B}_{\sigma}(0)} \|\nabla^2 f_{\mathrm{cl}, D, i}(z)\|_2$ denotes a uniform upper bound on the Hessian spectral norm for the $i$-th component of $f_{\mathrm{cl}, D}$ over $\mathcal{B}_{\sigma}(0)$. Recalling that the nominal closed-loop dynamics are given by $f_{\mathrm{cl}}(z) := Z_{1} \Phi_{0}^{\dagger} \Phi(z,\psi(z))$, $H_{i, D}$ can be conservatively upper-bounded by $\bar{H}_{i, D}$:
	\begin{equation} \label{eq:hessian_conservative_bound}
		\begin{aligned}
			\bar{H}_{i, D} := \; & \max_{z \in \mathcal{B}_{\sigma}(0)} \|\nabla^2 f_{\mathrm{cl}, i}(z)\|_2 + \|\Delta_z\|_2 \|\Phi_0^{\dagger}\|_2 \\
			& \quad \times \sum_{j=1}^s \max_{z \in \mathcal{B}_{\sigma}(0)} \|\nabla^2 \Phi_j(z, \psi(z))\|_2.
		\end{aligned}
	\end{equation}
	Let $K_{g, D} := \sum_{i=1}^n \|P_{i,\cdot}\|_2 \bar{H}_{i, D}$. Combining \eqref{eq:robust_quadratic_bound} and \eqref{eq:robust_higher_order_cross_term_bound} yields
	\begin{equation} \label{eq:robust_v_dot_final_bound}
		\dot{V}_{\text{loc}}(z) \le -\left( \gamma - K_{g, D} \|z\|_2 \right) \|z\|_2^2, \quad \forall D_{z,0} \in \mathcal{D}_z.
	\end{equation}
	Consequently, $\dot{V}_{\text{loc}}(z) < 0$ holds for all $z \in \mathcal{B}_{\sigma_{\text{loc}, D}}(0) \setminus \{0\}$ across the entire set of data-consistent dynamics, with $\sigma_{\text{loc}, D} := \frac{\gamma}{K_{g, D}}$. Analogous to the convergence analysis in Section~\ref{subsec:stability_formal_verification}, the set inclusion $\mathcal{B}_\sigma(0) \subseteq \mathcal{R}_{\text{loc}}$ serves as a sufficient condition for robust asymptotic convergence to the origin, which is formulated as
	\begin{equation} \label{eq:robust_sigma_threshold_bound}
		\sigma < \sigma_{\text{threshold}, D} := \sigma_{\text{loc}, D} \sqrt{\frac{\lambda_{\min}(P)}{\lambda_{\max}(P)}}.
	\end{equation}
	Together with the SMT verification over $\mathcal{X}_{\text{SMT}}$, condition \eqref{eq:robust_sigma_threshold_bound} completes the formal stability certification for the closed-loop system. Conversely, if \eqref{eq:robust_sigma_threshold_bound} is violated, the candidate controller is rejected, and the algorithm reverts to the training loop to re-synthesize $\psi(z)$ and $V_z(z)$.
	
	\begin{rem} \label{rem:robust_sdp_objective_and_sampling}
		Analogous to \eqref{eq:lmi_sdp_opt}, the SDP \eqref{eq:robust_lmi_sdp_opt} inherently optimizes $\sigma_{\text{threshold}, D}$. Furthermore, to compute $\bar{H}_{i, D}$ ($i = 1, \dots, n$) via \eqref{eq:hessian_conservative_bound}, a set of state samples $\{z_k\}_{k=1}^{\bar{N}}$ is densely drawn within $\mathcal{B}_{\sigma}(0)$. At each sample point, $\|\nabla^2 f_{\mathrm{cl}, i}(z_k)\|_2$ and $\|\nabla^2 \Phi_j(z_k, \psi(z_k))\|_2$ ($j = 1, \dots, s$) are evaluated using PyTorch. By taking the maximums over all evaluated samples and incorporating $\|\Delta_z\|_2$ along with $\|\Phi_0^{\dagger}\|_2$, the bounds $\bar{H}_{i, D}$ are determined. \hfill $\blacksquare$
	\end{rem}
	
	\begin{example}
		Consider the dynamics of an inverted pendulum:
		\begin{equation} \label{eq:inverted_pendulum}
			\begin{aligned}
				\dot{x}_1 &= x_2, \\
				\dot{x}_2 &= -\frac{\zeta}{m \ell^2} x_2 + \frac{g}{\ell} \sin x_1 + \frac{1}{m \ell^2} u,
			\end{aligned}
		\end{equation}
		where $x_1$ and $x_2$ denote the angular position and velocity, respectively, and $u$ is the applied torque. Note that the origin, representing the upright position, is an unstable equilibrium, and the control objective is to stabilize the closed-loop system at the origin.

		Let $\mathcal{X} = \{x \in \mathbb{R}^2 \mid -b_i \le x_i \le b_i, \; i = 1, 2\}$ denote the target stabilization domain, and $S = \operatorname{diag}(\frac{1}{b_1}, \frac{1}{b_2})$ be the state scaling matrix. By expressing \eqref{eq:inverted_pendulum} in terms of the normalized state $z = Sx$, the transformed dynamics governing $z$ are formulated as:
		\begin{equation} \label{eq:normalized_pendulum_dynamics}
			\begin{aligned}
				\dot{z}_1 &= \frac{b_2}{b_1} z_2, \\
				\dot{z}_2 &= -\frac{\zeta}{m \ell^2} z_2 + \frac{g}{\ell b_2} \sin(b_1 z_1) + \frac{1}{m \ell^2 b_2} u.
			\end{aligned}
		\end{equation}
		Accordingly, we construct the normalized function library as $\Phi(z, u) = [z_1, z_2, \sin(b_1 z_1), u]^\top$.
		
		In this example, we evaluate the proposed robust neural control framework subject to identification data perturbations. As the first step, we collect the dataset $\mathbb{D}$ from \eqref{eq:inverted_pendulum} with physical parameters $m = 0.1$, $\ell = 2.0$, $g = 9.81$, and $\Gamma= 0.01$. The bounds defining the domain $\mathcal{X}$ are set to $b_1 = 3.14$ and $b_2 = 2.0$. Specifically, we collect $T = 50$ samples by uniformly sampling the control input $u \in [-2, 2]$ and the state $x \in \mathcal{X}$. The collected samples are corrupted by additive perturbations $d(t_i)$ satisfying $\|d(t_i)\|_{\infty} \le \rho_d$ with $\rho_d = 0.01$. Accordingly, the scaled perturbation bound matrix is computed as $\Delta_z = S \sqrt{nT} \rho_d I_n$. The scalar parameter $\theta_{\alpha}$ is initialized to $-2.3$ (corresponding to $\alpha \approx 0.1$) and optimized jointly with the network parameters.
		
		For network configuration, both $V_z(z)$ and $\psi(z)$ are parameterized using feedforward neural networks with one hidden layer of $50$ neurons and $\tanh$ activation functions. The bias-elimination architecture is implemented to strictly enforce the equilibrium constraints $V_z(0) = 0$ and $\psi(0) = 0$. Prior to the joint synthesis phase, a warm-up stage is executed wherein $V_z(z)$ is pre-trained to approximate $V_{\text{target}}(z) = \|z\|_2^2$. Regarding the loss function, we set the convergence margin $\varepsilon = 0.1$ and the penalty hyperparameters to $\gamma_1 = \gamma_2 = 10$. During joint synthesis, the network parameters are optimized via the Adam optimizer with a learning rate of $5 \times 10^{-4}$.
		
		The training phase begins with the construction of the training dataset $\overline{\mathbb{D}}$. Initially, $\overline{\mathbb{D}}$ is populated with $50$ batches of samples, each with a batch size of $300$, generated via uniform state sampling across the domain $\mathcal{X}_z$. Let the exclusion ball be of radius $\sigma = 0.05$, and thus the SMT verification domain is $\mathcal{X}_{\text{SMT}} = \{z \in \mathbb{R}^2 \mid z \in [-1, 1]^2 \setminus \mathcal{B}_{\sigma}(0)\}$. Within each training loop, the counterexample search is executed via the PGD algorithm. Both the hybrid sampling strategy and the PGD algorithm share the same configuration as detailed in Example \ref{ex:lorenz}. For formal stability verification, we employ the SMT solver dReal to validate the normalized counterpart of the robust Lyapunov conditions \eqref{eq:robust_stability} over $\mathcal{X}_{\text{SMT}}$ with a precision tolerance of $\delta = 2 \times 10^{-5}$.
		
		In the numerical evaluation, dReal successfully certifies that the robust Lyapunov conditions hold universally across $\mathcal{X}_{\text{SMT}}$. Upon convergence, the optimized parameter yields $\alpha^* = e^{\theta_{\alpha}^*} = 0.2489$. Furthermore, regarding the local domain $\mathcal{B}_\sigma(0)$, evaluating the Hessian bounds $\bar{H}_{i, D}$ over $\mathcal{B}_{\sigma}(0)$ determines a threshold $\sigma_{\text{threshold}, D} = 0.2118$ that satisfies $\sigma < \sigma_{\text{threshold}, D}$, thereby ensuring robust asymptotic convergence under data perturbations. By Proposition \ref{prop:equivalence}, the synthesized neural controller $u = \psi(Sx)$ is formally guaranteed to stabilize the inverted pendulum at the upright equilibrium position.
		
		Let
		\begin{equation} \label{eq:robust_dissipation}
			\begin{aligned}
				\mathcal{W}(x) :=&\, 2\frac{\partial V_z(Sx)}{\partial x} X_{1}Z_{0}^{\dagger} Z(x, \psi(Sx)) \\
				& + \alpha^* \frac{\partial V_z(Sx)}{\partial x} \Delta \Delta^\top \left(\frac{\partial V_z(Sx)}{\partial x}\right)^\top \\
				& + \frac{1}{\alpha^*} Z(x,\psi(Sx))^\top (Z_0^\dagger)^\top Z_0^\dagger Z(x,\psi(Sx)).
			\end{aligned}
		\end{equation}
		For the closed-loop dynamics, we numerically determine the set $\mathcal{H}_{\mathcal{W}} := \{x \in \mathbb{R}^2 \mid \mathcal{W}(x) < 0\}$, over which the Lyapunov function $V_z(Sx)$ decreases, and any sub-level set $\mathcal{R}_c = \{x \in \mathbb{R}^2 \mid V_z(Sx) \le c\}$ with $c > 0$ contained within $\mathcal{H}_{\mathcal{W}} \cup \{0\}$ yields an estimate of the ROA for the closed-loop system. The set $\mathcal{H}_{\mathcal{W}}$, a sublevel set of $V_z(Sx)$, and the domain $\mathcal{X}$ are shown in Figure~\ref{fig:pendulum_roa_noise}. Notably, the green set is larger than any sub-level set of $V_z(Sx)$ contained in $\mathcal{X}$.
		
		Note that since $\Delta = \sqrt{nT} \rho_d I_n$, the uncertainty bound increases with the number of identification samples $T$. As a result, a larger $T$ introduces greater conservatism, making it more challenging to synthesize $\phi(x)$ and $V(x)$ that satisfy \eqref{eq:robust_stability}. In our experiments, robust stability can be certified up to $T = 1000$, which validates the effectiveness of the proposed approach. Furthermore, we observe that $\sigma_{\text{threshold}, D}$ decreases as $T$ increases, which is consistent with the fact that $K_{g, D}$ increases with $T$.
		
		\begin{figure}[ht!]
			\centering
			\includegraphics[width=0.88\columnwidth]{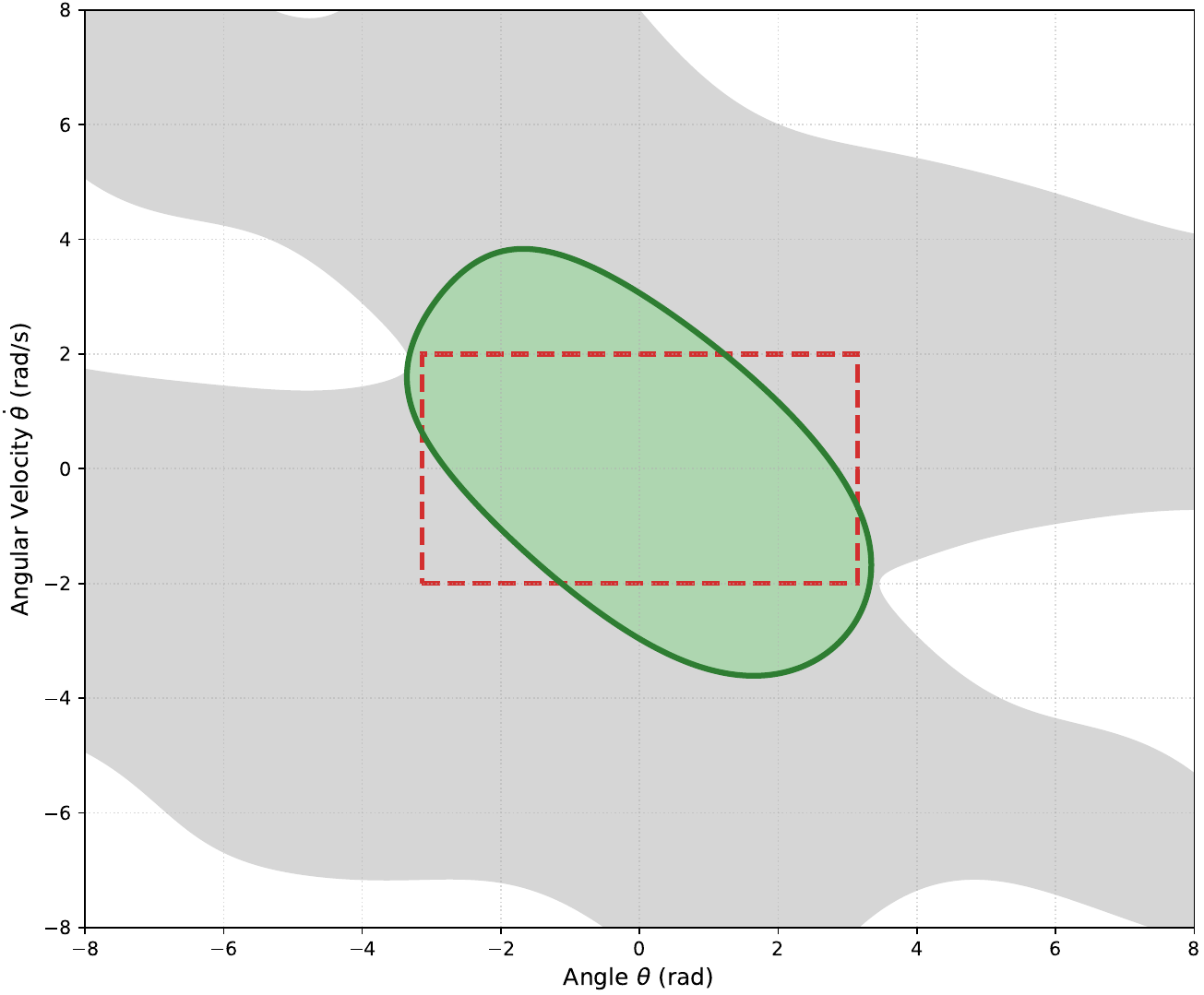}
			\caption{The grey set represents the set $\mathcal{H}_{\mathcal{W}}$, the green set is a sub-level set $\mathcal{R}_c$ contained in $\mathcal{H}_{\mathcal{W}} \cup \{0\}$ with $c = 3.5177$, and the red dashed box delineates the verified domain $\mathcal{X}$.}
			\label{fig:pendulum_roa_noise}
		\end{figure}
	\end{example}
	
	\section{Conclusion} \label{sec:conclusion}
	We have introduced a data-driven framework for the joint synthesis and formal verification of neural feedback controllers and neural Lyapunov functions for unknown nonlinear systems. We adopt an indirect control paradigm, wherein the system model is first identified from offline data, and the neural feedback law and Lyapunov function are subsequently synthesized based on the identified dynamics. To accommodate physical actuator limits, a hard saturation architecture is incorporated to enforce input constraints by construction. Furthermore, a robust extension is developed to ensure stability across all data-consistent models under identification data perturbations. Future research will focus on mitigating computational complexity for high-dimensional systems, as well as extending the proposed framework to output-feedback control.

	
	\bibliographystyle{unsrt} 
	\bibliography{OMAS} 

%
%
	
	\appendix
\end{document}